\documentclass[letterpaper, conference, twocolumn, 10pt]{ieeeconf} 

\usepackage{amsfonts, amsmath, amstext, amssymb}
\usepackage{color}
\usepackage{enumerate}
\usepackage{algorithm,algorithmic}
\usepackage{graphicx,subfigure}

\renewcommand{\Re}{\mathbb{R}}
\newcommand{\N}{\mathbb{N}}
\newcommand{\U}{\mathcal{U}}
\newcommand{\Y}{\mathcal{Y}}
\newcommand{\X}{\mathcal{X}}
\newcommand{\OCal}{\mathcal{O}}
\newcommand{\G}{\mathcal{G}}
\newcommand{\I}{\mathcal{I}}

\newcommand{\V}{\mathcal{V}}
\newcommand{\E}{\mathcal{E}}
\newcommand{\B}{\mathcal{B}}
\newcommand{\intnew}{\mathrm{int}}

\newcommand{\R}{\mathcal{R}}

\DeclareMathOperator*{\maximize}{maximize}
\DeclareMathOperator*{\subjto}{subject~to}

\newtheorem{problem}{Problem}
\newtheorem{theorem}{Theorem}
\newtheorem{proposition}{Proposition}
\newtheorem{corollary}{Corollary}

\title{\Huge Path Planning using Positive Invariant Sets} 
\author{Claus Danielson, Avishai Weiss, Karl Berntorp, and Stefano {Di Cairano} %
\thanks{C. Danielson, A. Weiss, K. Berntorp, and S. {Di Cairano} are with Mitsubishi Electric Research Laboratories, Cambridge MA, USA}
\thanks{e-mail: {\tt danielson@merl.com}}
} 

\IEEEoverridecommandlockouts

\begin{document}

\maketitle

\begin{abstract}
We present an algorithm for steering the output of a linear system from a feasible initial condition to a desired target position, while satisfying input constraints and non-convex output constraints.  The system input is generated by a collection of local linear state-feedback controllers.  The path-planning algorithm selects the appropriate local controller using a graph search, where the nodes of the graph are the local controllers and the edges of the graph indicate when it is possible to transition from one local controller to another without violating input or output constraints.  We present two methods for computing the local controllers.  The first uses a fixed-gain controller and scales its positive invariant set to satisfy the input and output constraints.  We provide a linear program for determining the scale-factor and a condition for when the linear program has a closed-form solution.  The second method designs the local controllers using a semi-definite program that maximizes the volume of the positive invariant set that satisfies state and input constraints.  We demonstrate our path-planning algorithm on docking of a spacecraft.  The semi-definite programming based control design has better performance but requires more computation.  
\end{abstract}

\section{Introduction}

In path-planning, the goal is to generate a trajectory between an initial state and a target state while avoiding obstacle collision. This is a computationally difficult problem \cite{reif1979complexity}. Recently, there has been much work on {sampling-based} motion planning, where the search space of possible trajectories is reduced to a graph search amongst randomly selected samples \cite{lavalle2006planning, lavalle2001randomized, karaman2011sampling}. Sampling-based planners have been applied to humanoid robots \cite{kuffner2002dynamically}, autonomous driving \cite{leonard2008perception}, robotic manipulators \cite{perez2011asymptotically}, and spacecraft motion problems \cite{frazzoli2003quasi, starek2016real}.

There are still many practical challenges in applying these methods to systems with complex dynamics in high dimensional spaces, since the algorithms must account for the dynamics and constraints of the system \cite{karaman2011sampling}. Instead, traditional sampling-based methods solve a {two-point boundary-value problem}, for each edge of the search graph, to find a feasible input and state trajectory.  This \emph{a posteriori} approach to satisfying the constraints and dynamics is computationally challenging \cite{vinter2010optimal}.  Furthermore, for systems subject to unmeasured disturbances, model uncertainty, and unmodeled dynamics, there is no guarantee that the resulting open-loop trajectories will satisfy the system output constraints. 

In \cite{Weiss2014} a path planner, based on the use of positively invariant sets, was developed for the spacecraft obstacle avoidance problem.  Unlike the aforementioned planners, the focus is not on optimal trajectory planning or precise reference tracking.  Instead, the planner in  \cite{Weiss2014} implicitly finds a state trajectory from the initial state to the target state that explicitly satisfies the system dynamics and constraints.  The algorithm uses a graph to switch between a collection of local feedback controllers.  Positive invariant sets are used to determine when it is possible to transition from one controller to another, without violating input or output constraints.  A set is positive invariant if any closed-loop state trajectory that starts inside the set remains in the set for all future times.  By choosing positive invariant subsets of the input and output constraint sets, it is possible to guarantee that the closed-loop trajectories satisfy the constraints.  Using a graph and simple feedback controllers to generate the control input ensures that the algorithm has low computational complexity.  This concept can be extended to systems with set-bounded disturbances and differential inclusion model uncertainty, by using robust positive invariant sets \cite{Kolmanovsky1998, Kerrigan2000}.  By sampling feedback controllers, as opposed to points in the output space, the planner inherently accounts for the dynamics of the system and produces constraint feasible trajectories.  A similar idea is used in \cite{arslan2016sampling, McConley2000, Blanchini2004}.  

In this paper, we extend the path-planning algorithm presented in \cite{Weiss2014}.  Our analysis provides a sufficient condition for the existence of a path from the initial output to the target output that satisfies the system dynamics and constraints.    Furthermore we provide conditions, under which, our path-planning algorithm solves the path-planning problem.   We introduce two methods for computing local controllers and their associated positive invariant sets.   The first method follows the idea from \cite{Weiss2014}.  It uses a fixed-gain controller and scales its positive invariant set to guarantee input and output constraint satisfaction. In this paper, however, we consider the output constraints as a union of convex sets that represents the free space, and provide a linear program (LP) for determining the scale-factor. We also provide a condition for when this linear program has a closed-form solution. In the second method, we design the local controllers using a semi-definite program (SDP) that maximizes the volume of the positive invariant set that satisfies state and input constraints.  This approach increases the number of edges in the controller graph and can potential provide better performance, albeit at the expense of solving a computationally burdensome SDP in place of a simple LP.

The paper begins by describing the path-planning problem in Section~\ref{sec:path-planning-problem} along with a condition of the existence of a solution. Section~\ref{sec:path-planning-algorithm} details our path-planning algorithm for solving the path-planning problem. Two new methods for designing the local controllers and associated positive invariant sets are presented in Section~\ref{sec:control-design}. Finally, we demonstrate our path-planning algorithm on the docking of a spacecraft in Section~\ref{sec:example}.

\subsection{Notation and Definitions}

A ball $\B(c,r) = \{ x \in \Re^n : \| x - c \|_2 \leq r \} \subseteq \Re^n$ is the set of all points $x \in \Re^n$ whose Euclidean distance from the center $c \in \Re^n$ is less than the radius $r \in \Re_+$.  A polytope $\X = \{ x \in \Re^n: H^j x \leq K^j \text{ for } j = 1,\dots,m \} \subseteq \Re^n$ is the intersection of a finite number of half-spaces.  A polytope $\X$ is called full-dimensional if it contains a ball $\B(c,r) \subseteq \X$ with positive radius $r >0$.  The Chebyshev radius of a full-dimensional polytope $\X$ is the radius $r>0$ of the largest ball $\B(c,r) \subseteq \X$ contained in $\X$.  A point $x \in \X$ is in the interior of $\X$ if there exists a radius $r>0$ such that the ball $\B(x,r) \subseteq \X$ is contained in the set $\X$.  The set of all points $x \in \X$ in the interior of $\X$ is denoted by $\intnew(\X)$.  The image of the set $\X \subseteq \Re^n$ through the matrix $A \in \Re^{n \times n}$ is $A \X = \{ A x : x \in \X\}$.  The Pontryagin difference of $\X \subseteq \Re^n$ and $\Y \subseteq \Re^n$ is the set $\X \ominus \Y = \{ x : x+y \in \X, \forall y \in \Y\}$.  

A set $\OCal$ is positive invariant (PI) for the autonomous system $x(t+1) = f(x(t))$ if for every state $x(t) \in \OCal$ we have $x(t+1) = f(x(t)) \in \OCal$.  If $V(x)$ is a Lyapunov function for the stable autonomous system $x(t+1) = f(x(t))$, then any level-set $\OCal = \{ x \in \Re^n : V(x) \leq l \}$ is a PI set since $V(f(x)) \leq V(x)$.  

For a system $x(t+1) = f(x(t),u(t))$ with input $u \in \U$ and state $x \in \X$ constraints, the set of states $x \in \R_N(\bar{x})$ reachable from some initial state $\bar{x} \in \X$ in $N \in \N$ steps is defined recursively as $\R_0(\bar{x}) = \{\bar{x}\}$ and
\begin{align*}
\R_{k+1} = \big\{ x \in \X : &x = f(z,u),  z \in \R_{k}, u \in \U \big\} \cup \R_{k}
\end{align*} 
for $k = 0,\dots,N-1$, where $\R_{k} = \R_{k}(\bar{x})$.  The infinite-horizon reachable set is the limit $\R_\infty(\bar{x}) = \lim_{N\rightarrow \infty} \bigcup_{k=0}^N \R_k(\bar{x})$.  If the system $x(t+1) = f(x(t),u(t))$ is locally controllable, $f$ is continuous, and the input set $\U$ and state set $\X$ are full-dimensional and contain a neighborhood of the origin, then there exists an time $N \in \N$ such that the reachable set $\R_N(\bar{x})$ is full-dimensional.  

A directed graph $\G = (\V,\E)$ is a set of vertices $\V$ together with a set of ordered pairs $\E \subseteq \V \times \V$ called edges.  Vertices $i,j \in \V$ are called adjacent if $(i,j) \in \E$ is an edge.  A path is a sequence of adjacent vertices.  A graph search is an algorithm for finding a path through a graph.  

\section{path-planning Algorithm}
\label{sec:path-planning}

In this section we define the path-planning problem and our algorithm for solving it.  Our algorithm extends the method presented in \cite{Weiss2014}.  

\subsection{Path-Planning Problem}
\label{sec:path-planning-problem}

Consider the following discrete-time linear system
\begin{subequations}
\label{eq:system}
\begin{align}
x(t+1) &= A x(t) + B u(t) \\
y(t)     &= C x(t)
\end{align}
\end{subequations}
where $x(t) \in \Re^{n_x}$ is the state, $u(t) \in \Re^{n_u}$ is the control input, and $y(t) \in \Re^{n_y}$ is the output.  The pair $(A,B)$ is assumed to be controllable and $\mathrm{rank}(C) = n_y$.  The input $u(t)$ and output $y(t)$ are subject to constraints
\begin{align*}
u(t) \in \U  \text{ and } y(t) \in \Y 
\end{align*}
where the input set $\U \subset \Re^{n_u}$ is a full-dimensional compact polytope
\begin{subequations}
\label{eq:constraints}
\begin{align}
\label{eq:input-set}
\U = \big\{ u : H_{u}^j u \leq K_{u}^j \text{ for } j = 1,\dots,m_{u} \big\}.
\end{align}  
The output set $\Y \subseteq \Re^{n_y}$ is generally non-convex, but can be described as the union of convex sets
\begin{align}
\label{eq:output-set}
\Y = \displaystyle{\bigcup}_{k \in \I_\Y} \Y_k
\end{align}
where the index set $\I_\Y$ is finite $| \I_\Y| < \infty$ and each component set $\Y_k \subseteq \Re^{n_y}$ is a full-dimensional compact polytope
\begin{align}
\label{eq:output-subsets}
\Y_k = \big\{ y : H_{y_k}^j y \leq K_{y_k}^j \text{ for } j = 1,\dots,m_{y_k}\big\}.
\end{align}  
\end{subequations}
We assume that each output $\bar{y} \in \Y$ corresponds to a feasible equilibrium of the system~\eqref{eq:system} i.e. for each output $\bar{y} \in \Y$ there exists a state $\bar{x} \in \Re^{n_x}$ and feasible input $\bar{u} \in \intnew(\U)$ such that 
\begin{align}
\label{eq:equilibrium}
\begin{bmatrix} A - I & B \\ C & 0 \end{bmatrix}
\begin{bmatrix} \bar{x} \\ \bar{u} \end{bmatrix}
=
\begin{bmatrix} 0 \\ \bar{y} \end{bmatrix}.
\end{align}
The equilibrium state $\bar{x}$ and input $\bar{u}$ may not be unique.

The objective of the path-planning problem is to drive the system output $y(t)$ from the feasible initial equilibrium output $y(0) = y_0 \in \intnew(\Y)$ to a desired target output $y(t) \rightarrow y_f \in \intnew(\Y)$.  The path-planning problem is formally stated below.

\begin{problem}
\label{prob:path-planning}
Find a feasible input trajectory $u(t) \in \U$ for $t \in \N$ that produces a feasible output trajectory $y(t) \in \Y$ that converges to the target output $\lim_{t \rightarrow \infty} y(t) = y_f$. \hfill $\square$
\end{problem}

Before detailing our algorithm for solving Problem~\ref{prob:path-planning}, we study the conditions for when a solution exists.  Consider the following graph $\G_\Y = (\I_\Y,\E_\Y)$ where the nodes of the graph are the indices $\I_\Y$ of the component sets $\Y_k$ that comprise $\Y$ in~\eqref{eq:output-set}.  An edge $(i,j) \in \E_\Y$ connects nodes $i,j \in \I_\Y$ if the intersection of the sets $\Y_i$ and $\Y_j$ has a non-empty interior $\intnew(\Y_i \cap \Y_j) \neq \varnothing$.  The following theorem provides a sufficient condition for the existence of a solution to Problem~\ref{prob:path-planning}.  

\begin{theorem}
\label{thm:solution-existence}
If the graph $\G_\Y$ has a path from a node $\Y_0$ containing the initial equilibrium output $y_0$ to a node $\Y_f$ containing the target output $y_f$ then Problem \ref{prob:path-planning} can be solved.  
\end{theorem}

\begin{proof}
First we show that the connectivity of the graph $\G_\Y$ implies that the non-convex set $\Y$ contains, strictly in its interior, a path from the initial equilibrium output $y_0$ to the target output $y_f$.  For notational simplicity, we will assume that the path through the graph $\G_\Y$ from a node containing $y_0$ to a node $y_f$ is labeled $\Y_1,\dots,\Y_{f-1}$.  By definition of the graph edges, the intersection of sets $\Y_i$ and $\Y_{i+1}$ for $i = 0,\dots,f-1$ has a non-empty interior $\intnew(\Y_i \cap \Y_{i+1}) \neq \varnothing$.  Define $y_{i}$ for $i=1,\dots,f-1$ as the Chebyshev center of the set $\Y_i \cap \Y_{i+1}$.  By definition of the Chebyshev ball, the points $y_{i}$ for $i = 1,\dots,f-1$ are in the interiors of the sets $\Y_i$ and $\Y_{i+1}$.  Furthermore, by assumption the points $y_0 \in \intnew(\Y)$ and $y_f \in \intnew(\Y)$ are contained in the interior of $\Y$.  Without loss of generality we can assume that $y_0 \in \intnew(\Y_1)$ is contained in the interior of $\Y_1$ and $y_f \in \intnew(\Y_f)$ is contained in the interior of $\Y_f$.  Thus $y_i \in \intnew(\Y)$ for $i = 0,\dots,f$.  By convexity, this implies that the line segments $[y_i,y_{i+1}] \subset \intnew(\Y)$ for $i = 0,\dots,f-1$ are contained strictly inside $\Y$.  Furthermore these line segments form a path of outputs from the initial equilibrium output $y_0$ to the target output $y_f$.  

Next we show that the endpoint $y_{i+1}$ of each line segment $[y_i,y_{i+1}]$ for $i = 0,\dots,f-1$ is reachable in finite time.  Since the dynamics are continuous, there exists a connected path of equilibrium states $\bar{x}_\lambda = f_{\bar{x}}(\lambda)$ and inputs $\bar{u}_\lambda = f_{\bar{u}}(\lambda)$ for each output $\bar{y}_\lambda = (1-\lambda) y_i + \lambda y_{i+1}$ along the line segment $[y_i,y_{i+1}]$ where $\lambda \in [0,1]$.  The input set $\U$ and state set $\X = \{ x : C x \in \Y \}$ are full-dimensional since $\Y$ is full-dimensional and $\mathrm{rank}(C) = n_y$.  Thus for each $\lambda \in [0,1]$ the reachable set $\R_\infty(\bar{x}_\lambda)$ is full-dimensional since the dynamics are locally controllable and continuous, and the sets $\X\ominus\{\bar{x}_\lambda\}$ and $\U\ominus\{\bar{u}_\lambda\}$ are full-dimensional and contain the origin in their interiors.  Hence the collection of reachable sets $\R_\infty(\bar{x}_\lambda)$ cover the set of equilibrium state $\{ \bar{x}_\lambda : \lambda \in [0,1]\}$.  Let us define the following collection of restricted reachable sets
\begin{align*}
\R_\lambda = \big\{ x \in \intnew\big(\R_\infty(\bar{x}_\lambda)\big) : (y_{i+1}-y_i)^T (Cx-\bar{y}_\lambda) > 0 \big\}. 
\end{align*}
This is the set of states $x \in \intnew\big(\R_\infty(\bar{x}_\lambda)\big)$ reachable from $\bar{x}_\lambda$ in finite time-steps whose output $y = Cx$ is ``closer'' to $y_{i+1}$ than $\bar{y}_\lambda$.  These sets $\R_\lambda$ are open since they are the intersect of two open sets: $\intnew\big(\R_\infty(\bar{x}_\lambda)\big)$ and $\{ x : (y_{i+1}-y_i)^T (Cx-\bar{y}_\lambda) > 0 \}$.  The sets $\R_\lambda$ are full-dimensional since $\bar{x}_\lambda \in \intnew\big(\R_\infty(\bar{x}_\lambda)\big)$.  Furthermore the infinite collection of sets 
\begin{align*}
\R_\infty(\bar{x}(0)) \text{ and } \R_\lambda \text{ for } \lambda \in [0,1] 
\end{align*}
covers the equilibrium states $\{ \bar{x}_\lambda : \lambda \in [0,1]\}$ since $\bar{x}_\lambda = f_{\bar{x}}(\lambda)$ depends continuously on $\lambda \in [0,1]$ and $\R_\lambda$ is full-dimensional and intersects every neighborhood of $\bar{x}_\lambda$.  

The set of equilibrium states $\{ \bar{x}_\lambda: \lambda \in [0,1] \}$ is compact since it is the continuous image of the compact set $[0,1]$.  Thus we can select a finite subcover $\R_\infty(\bar{x}_0),\R_{\lambda_0}, \dots, \R_{\lambda_L}$ of the equilibrium set.  Without loss of generality, we assume that $0 = \lambda_0 < \cdots < \lambda_L = 1$.  Since the sets $\R_{\lambda_0}, \dots, \R_{\lambda_L}$ cover the set of equilibrium states $\{ \bar{x}_\lambda : \lambda \in [0,1]\}$, there must exist $\R_{\lambda_i}$ that contains $\bar{x}_{\lambda_j}$ for each $j = 1,\dots,L$.  By definition of the restricted reachable sets $\R_\lambda$, we have $\lambda_i < \lambda_j$.  Thus the equilibrium state $\bar{x}_{\lambda_j}$ is reachable from $\bar{x}_{\lambda_i}$ where $\lambda_i < \lambda_j$ since $\R_{\lambda_i} \subset \R_\infty(\bar{x}_{\lambda_i})$.  By induction on $j$, this implies that the equilibrium state $\bar{x}_{\lambda_L} = \bar{x}_1$ corresponding to the endpoint $y_{i+1} = \bar{y}_{1}$ is reachable in finite time from the equilibrium state $\bar{x}_{\lambda_0} = \bar{x}_0$ corresponding to the begin point $y_{i} = \bar{y}_{0}$ of the line segment $[y_i,y_{i+1}]$.  

We conclude that the target output $y_f$ is reachable in a finite number of time-steps since there are a finite number of linear segments $[y_i,y_{i+1}]$ for $i=0,\dots,f-1$ and each endpoint $y_{i+1}$ is reachable in finite time.  Furthermore by the definition of the reachable sets $\R_\infty(\bar{x})$, the control inputs $u(t)$ needed to reach each end point $y_{i+1}$ are feasible $u(t) \in \U$ and the resulting state trajectory $x(t)$ is feasible $x(t) \in \X = \{ x : C x \in \Y \}$.  Therefore a solution to Problem \ref{prob:path-planning} exists.  
\end{proof}

The proof of Theorem \ref{thm:solution-existence} can be found in \cite{FullPaper}.  The proof does not rely on the linearity of the dynamics~\eqref{eq:system} nor the polytopic union structure of the constraints~\eqref{eq:constraints}.  Rather the proof uses the fact that the dynamics $f(x,u) = Ax + Bu$ are continuous and locally controllable about each equilibrium~\eqref{eq:equilibrium}, and that the constraints are full-dimensional.  Thus Theorem \ref{thm:solution-existence} can be extended to non-linear systems with more complicated constraints.  However in this paper we consider linear dynamics with polytopic union constraints, since our path-planning algorithm exploits these properties.

\subsection{Path-Planning Algorithm}
\label{sec:path-planning-algorithm}

In this section we present our algorithm for solving Problem~\ref{prob:path-planning}.  This algorithm extends the concept first presented in \cite{Weiss2014}.  

The control input $u(t)$ is selected from a collection of local linear state-feedback controllers of the form 
\begin{align}
\label{eq:local-control}
u_i = F_i (x -\bar{x}_i) + \bar{u}_i
\end{align}
for $i \in \I$, where $\I$ is the index set of the local controllers, and $\bar{x}_i$ and $\bar{u}_i \in \intnew(\U)$ are an equilibrium~\eqref{eq:equilibrium} state and input pair corresponding to the output $\bar{y}_i \in \intnew(\Y)$.  We assume that the local controller~\eqref{eq:local-control} asymptotically stabilizes its equilibrium point $\bar{x}_i$ i.e. the matrix $A+BF_i$ is Schur.  Thus each controller~\eqref{eq:local-control} has an associated ellipsoidal positive invariant (PI) set 
\begin{align}
\label{eq:pi-sets}
\OCal_i = \big\{ x \in \Re^{n_x} : (x-\bar{x}_i)^T P_i  (x-\bar{x}_i) \leq 1 \big\}
\end{align}
where $V(x) = (x-\bar{x}_i)^T P_i  (x-\bar{x}_i)$ is a quadratic Lyapunov function for the controller~\eqref{eq:local-control}.  We assume that the Lyapunov matrix $P_i \in \Re^{n_x \times n_x}$ is scaled such that for every state $x \in \OCal_i$ in the PI set $\OCal_i \subset \Re^{n_x}$, the output $y = C x \in \Y$ and input $u = F_i(x-\bar{x}_i) + \bar{u}_i \in \U$ are feasible.  Such a scaling is possible since $\bar{u}_i \in \intnew(\U)$ and $\bar{y}_i \in \intnew(\Y)$, and the set $\{ x : C x \in \Y \}$ is full-dimensional, since $\Y$ is full-dimensional and $\mathrm{rank}(C) = n_y$.  

The path-planning algorithm selects which local controller~\eqref{eq:local-control} to use based on a graph search.  The vertices $\I$ of the graph $\G = (\I,\E)$ are the indices of the local feedback controllers~\eqref{eq:local-control}.  Two controllers $i,j \in \I$ are connected by an edge $(i,j) \in \E$ if the equilibrium state $\bar{x}_i$ of controller $i \in \I$ is inside $\bar{x}_i \in \intnew(\OCal_j)$ the PI set $\OCal_j$ of controller $j \in \I$.  The presence of an edge $(i,j) \in \E$ means that it is possible to safely transition from controller $i \in \I$ to controller $j \in \I$ once the state $x(t)$ reaches a neighborhood of the equilibrium $\bar{x}_i$.  

We make the following assumptions about the controller graph $\G$:
\begin{enumerate}[\quad 1.]
\item[A1.] The set of controllers $\I$ contains at least one controller $u_f = F_f (x - \bar{x}_f) + \bar{u}_f$ corresponding to the target output $y_f = C \bar{x}_f \in \Y$.  
\item[A2.] The initial state $x(0) = x_0$ of the system~\eqref{eq:system} is contained in the PI set $\OCal_i$ of at least one controller $i \in \I$.  
\item[A3.] There exists a path through the graph $\G = (\I,\E)$ from a node containing the initial state $x(0) = x_0$ to a node corresponding to the target output $y_f$.  
\end{enumerate}

Our path-planning algorithm is summarized in Algorithm~\ref{alg:path-planning}.  Offline, the path-planning algorithm searches the controller graph $\G = (\I,\E)$ for a sequence $\sigma_0,\dots,\sigma_N \in \I$ of local controllers~\eqref{eq:local-control} from a node $\sigma_0$, whose PI set $\OCal_{\sigma_0}$ contains the inital state $x(0) \in \OCal_{\sigma_0}$ to a node $\sigma_f$, whose PI set $\OCal_{\sigma_f}$ contains an equilibrium state $\bar{x}_f$ corresponding to the target output $y_f$.  At each time instance $t \in \N$, the path planner uses the control input $u(t) = F_{\sigma(t)}(x(t) - \bar{x}_{\sigma(t)}) + \bar{u}_{\sigma(t)}$ where $\sigma(t)$ is the current node.  The controller node $\sigma(t)$ is updated $\sigma(t) = \sigma_{i+1}$ when the state $x(t)$ reaches the PI set $\OCal_{\sigma_{i+1}}$ for the next local controller $\sigma_{i+1}$. The initial node is $\sigma(t) = \sigma_0$.  

\begin{algorithm}
\caption{Path-Planning Algorithm}
\label{alg:path-planning}
\begin{algorithmic}[1]
\STATE initial local controller $\sigma(t) = \sigma_0$ 
\FOR{each time $t \in \N$}
\IF{$x(t) \in \OCal_{\sigma_{i+1}}$}
\STATE update local controller~\eqref{eq:local-control} node $\sigma(t) = \sigma_{i+1}$
\ELSE
\STATE use same local controller $\sigma(t) = \sigma_{i}$
\ENDIF
\STATE $u(t) =F_{\sigma(t)}(x(t) - \bar{x}_{\sigma(t)}) + \bar{u}_{\sigma(t)}$
\ENDFOR
\end{algorithmic}
\end{algorithm}

The following theorem shows that the path-planning algorithm satisfies the constraints and drives the system to the target output $y(t) \rightarrow y_f \in \Y$.  

\begin{theorem}
Algorithm \ref{alg:path-planning} solves Problem \ref{prob:path-planning}
\end{theorem}

\begin{proof}
First we note that the input trajectory $u(t) \in \U$ and resulting output trajectory $y(t) \in \Y$ are feasible for all $t \in \N$.  This follows from the fact that Algorithm \ref{alg:path-planning} transitions between PI sets $\OCal_i$ of the local controllers and the fact that the PI sets $\OCal_i$ are safe i.e. every state $x \in \OCal_i$ in the PI $\OCal_i$ set satisfies output $Cx \in \Y$ and input $F(x-\bar{x}_i) + \bar{u}_i \in \U$ constraints.  

Next we show that the output converges $y(t) \rightarrow y_f$ to the target output $y_f \in \Y$.  By assumption, the controller graph $\G$ contains a controller that asymptotically stabilizes an equilibrium $\bar{x}_f$ corresponding to the target output $y_f \in \Y$.  Therefore if the state $x(t)$ reaches the PI set $\OCal_f$ associated with this controller, then the output $y(t)$ will asymptotically converge to the target output $y(t) \rightarrow y_f$.  We now show that the state $x(t)$ reaches the set $\OCal_f$ after a finite number of time-steps.  

By assumption, the controller graph has a path $\sigma_0,\dots,\sigma_f \in \I$ from a node $\sigma_0 \in \I$ whose PI set $\OCal_{\sigma_0}$ contains the initial state $x(0) \in \OCal_{\sigma_0}$ to a node $\sigma_f \in \I$ whose PI set $\OCal_{\sigma_f}$ contains an equilibrium state $\bar{x}_{\sigma_f}$ corresponding to the target output $y_f = C\bar{x}_{\sigma_f}$.  Thus we have a sequence of equilibrium states $\bar{x}_{\sigma_0},\dots,\bar{x}_{\sigma_f}$ that the local controllers~\eqref{eq:local-control} stabilize.  We now show that the state $x(t)$ under the local controller $\sigma_{i} \in \I$ will reach the PI set $\OCal_{\sigma_{i+1}}$ of controller $\sigma_{i+1} \in \I$ in a finite number of time-steps.  

By definition of the edges $(\sigma_i,\sigma_{i+1}) \in \E$ of the controller graph $\G$, the equilibrium state $\bar{x}_{\sigma_i}$ is contained in the interior of the PI set $\OCal_{\sigma_{i+1}}$.  Thus there exists $\alpha \leq 1$ such that $\alpha \OCal_{\sigma_{i}} \subseteq \OCal_{\sigma_{i+1}}$ i.e. 
\begin{align*}
& \alpha (x(t)-\bar{x}_{\sigma_i})^T P_{\sigma_i} (x(t)-\bar{x}_{\sigma_i}) \\
& \qquad \leq (x(t)-\bar{x}_{\sigma_{i+1}})^T P_{\sigma_{i+1}} (x(t)-\bar{x}_{\sigma_{i+1}}).
\end{align*}  
Since the controller~\eqref{eq:local-control} asymptotically stabilizes $\bar{x}_{\sigma_i}$, there exists $\lambda$ such that 
\begin{align*}
(A+BF_{\sigma_i})^T P_i (A+BF_{\sigma_i}) \preceq \lambda P_i
\end{align*}
where $| \lambda | < 1$.  Since there exists $k_i \in \N$ such that $\lambda^{k_i} \leq \alpha$ we have
\begin{align*}
(x(t+k_i)-&\bar{x}_{\sigma_i})^T P_{\sigma_i} (x(t+k_i)-\bar{x}_{\sigma_i}) \\
& = (x(t)-\bar{x}_{\sigma_i})^T A_{cl}^{k_iT} P_i A_{cl}^{k_i} (x(t)-\bar{x}_{\sigma_i}) \\
& \leq \lambda^{k_i} (x(t)-\bar{x}_{\sigma_i})^T P_i (x(t)-\bar{x}_{\sigma_i}) \\
& \leq \alpha (x(t)-\bar{x}_{\sigma_i})^T P_i (x(t)-\bar{x}_{\sigma_i}) \\
& \leq  (x(t)-\bar{x}_{\sigma_{i+1}})^T P_{\sigma_{i+1}} (x(t)-\bar{x}_{\sigma_{i+1}})
\end{align*}
where $A_{cl} = A + BF_{\sigma_i}$ and 
\begin{align*}
x(t+1) 
&= A x(t) + B F_{\sigma_i} (x(t)-\bar{x}_{\sigma_i} )+ B \bar{u}_{\sigma_i} \\
&= (A + BF_{\sigma_i})(x(t)-\bar{x}_{\sigma_i} ) + \bar{x}_{\sigma_i}
\end{align*}
since $\bar{x}_{\sigma_i}$ and $\bar{u}_{\sigma_i}$ are an equilibrium state-input pair~\eqref{eq:equilibrium}.  Thus $x(t+k_i) \in \OCal_{\sigma_i+1}$ if $x(t) \in \OCal_{\sigma_i}$.  Since by assumption the initial state $x(0) \in \OCal_{\sigma_0}$ is inside the PI set $\OCal_{\sigma_0}$, we conclude by induction on $i$ that the state $x(t)$ reaches the set $\OCal_{\sigma_f}$ in a finite number of time-steps $k = k_0 + \cdots + k_f$ and thus the output $y(t)$ converges to the target output $y(t) \rightarrow y_f$.  
\end{proof}

Algorithm \ref{alg:path-planning} is a general path-planning algorithm.  There are three questions that must be addressed to implement this algorithm:
\begin{enumerate}[\quad 1.]
\item How do we design the local controllers~\eqref{eq:local-control} such that their PI sets~\eqref{eq:pi-sets} satisfy the input and output constraints?
\item How do we sample the output set $\Y$ such that the controller graph $G = (\I,\E)$ contains a path from a node $i \in \I$ whose PI $\OCal_i$ contains the initial state $x_0 \in \OCal_i$ to a node $f \in \I$ whose PI set $\OCal_f$ contains an equilibrium state $x_f \in \OCal_f$ corresponding to the target output $y_f \in \Y$.  
\item How do we weight the graph $G = (\I,\E)$ to provide good performance?
\end{enumerate}

This paper focuses on answering the first question.  

\section{Control Design}
\label{sec:control-design}

In this section we present two methods for designing the local controllers~\eqref{eq:local-control} with PI sets~\eqref{eq:pi-sets} that satisfy the input $F_i (\OCal_i-\bar{x}_i)+\bar{u}_i \subseteq \U$ and output constraints $C \OCal_i \subseteq \Y$.

\subsection{Fixed-Gain Controller with Scaled Invariant Set}
\label{sec:control-design-fixed}

Our first method uses a single feedback gain matrix $F_i = F$ for each $i \in \I$ of the local controller~\eqref{eq:local-control}.  Each of the local controllers~\eqref{eq:local-control} has a quadratic Lyapunov function of the form $V(x) = (x-\bar{x}_i)^T P (x-\bar{x}_i)$ that share a common Lyapunov matrix $P_i = P \succ 0$.  The PI sets of the local controller $i \in \I$ is the $\rho_i^2$ level-set of the local Lyapunov function 
\begin{align}
\label{eq:scaled-pi-sets}
\OCal_i(\rho_i) = \big\{ x \in \Re^{n_x} : (x-\bar{x}_i)^T P (x-\bar{x}_i) \leq \rho_i^2 \big\}.
\end{align}
To maximize the number of edges $(i,j) \in \E$ between the local controllers, we would like to maximize the volume of the PI set~\eqref{eq:scaled-pi-sets} by choosing the maximum level $\rho_i^2$ for which we can still satisfy the input $F (\OCal_i(\rho_i)-\bar{x}_i) + \bar{u}_i \subseteq \U$ and output $C \OCal_i(\rho_i) \subseteq \Y$ constraints.  Since the output set $\Y$ is generally non-convex, this is a non-convex problem.  However the output $y_i \in \Y = \bigcup \Y_k$ must be contained in at least one component set $\Y_k \subseteq \Y$ of $\Y$. Then we can solve the relaxed convex problem
\begin{subequations}
\label{eq:rho-max}
\begin{align}
\maximize 	~&~ 	\rho_i \\
\subjto 		~&~	F (\OCal_i(\rho_i)-\bar{x}_i)+\bar{u}_i \subseteq \U \\
			~&~	C \OCal_i(\rho_i) \subseteq \Y_k		
\end{align}
\end{subequations}
which finds the maximum level $\rho_i$ such that $C \OCal_i(\rho_i) \subseteq \Y_k$ is contained in the convex subset $\Y_k \subseteq \Y$ of $\Y$.  

If the system~\eqref{eq:system} has multiple equilibria for the output sample $\bar{y}_i \in \Y$ then the decision variables of problem~\eqref{eq:rho-max} are the level $\rho_i \in \Re$, the equilibrium state $\bar{x}_i \in \Re^{n_x}$, and input $\bar{u}_i \in \U$.  In this case, problem~\eqref{eq:rho-max} can be recast as a linear program.

\begin{proposition}
Problem~\eqref{eq:rho-max} is equivalent to the linear program
\begin{subequations}
\label{eq:rho-lp}
\begin{align}
\max_{\rho_i, \bar{x}_i, \bar{u}_i} ~&~ \rho_i \\
\mathrm{s.t.}	
~&~ H_{u}^j \bar{u}_i + \rho_i \| H_{u}^j F P^{-1/2} \|_2 \leq K_{u}^j \label{eq:rho-input} \\
~&~ H_{y_k}^j \bar{y}_i + \rho_i \| H_{y_k}^j C P^{-1/2} \|_2 \leq K_{y_k}^j \label{eq:rho-output} \\	
~&~ A \bar{x}_i + B \bar{u}_i = 0,~  C \bar{x}_i  = \bar{y}_i, \bar{u}_i \in \U \label{eq:rho-equilibrium}
\end{align}
\end{subequations}
where $(H_{u}^j,K_{u}^j)$ for $j = 1,\dots,m_u$ are the half-spaces of the input set~\eqref{eq:input-set} and $(H_{y_k}^j,K_{y_k}^j)$ are the half-spaces of the $k$-th output set~\eqref{eq:output-subsets}.
\end{proposition}

\begin{proof}
Problem~\eqref{eq:rho-max} can be rewritten as 
\begin{align*}
\max_{\rho_i, \bar{x}_i, \bar{u}_i} ~&~ \rho_i \\
\mathrm{s.t.}	~&~	F (x-\bar{x}_i)+\bar{u}_i \in \U  && \forall x \in \OCal(\rho_i) \\
			~&~	C x \in \Y_k && \forall x \in \OCal(\rho_i). 
\end{align*}
Performing the change of variables $e = P^{1/2}(x-\bar{x}_i)$ we obtain 
\begin{align*}
\max_{\rho_i, \bar{x}_i, \bar{u}_i} ~&~ \rho_i \\
\mathrm{s.t.}	~&~	F P^{-1/2} e + \bar{u}_i \in \U  && \forall e \in \mathcal{B}(0,\rho_i) \\
			~&~	C P^{-1/2} e + C \bar{x}_i \in \Y_k && \forall e \in \mathcal{B}(0,\rho_i) 
\end{align*}
where $\mathcal{B}(0,\rho_i) = \{ e : \| e \|_2 \leq \rho_i \}$ is the origin-centered ball of radius $\rho_i$.  Consider the $j$-th constraint of $\U$
\begin{align}
\label{pf:rho-input}
H_u^j F P^{-1/2} e + H_u^j \bar{u}_i \leq K_u^j  && \forall e \in \mathcal{B}(0,\rho_i).
\end{align}
Constraint~\eqref{pf:rho-input} is satisfied for all $e$ such that $\| e \| \leq \rho_i$ if and only if it is satisfied for (see \cite{Borrelli2009} Section~5.4.5)
\begin{align*}
e = \rho_i\tfrac{P^{-1/2}F^T (H_u^j)^T}{\|H_u^j F P^{-1/2}\|}.
\end{align*}
Therefore constraint~\eqref{pf:rho-input} can be replaced by~\eqref{eq:rho-input} for each $j = 1,\dots,m_u$.  Likewise we can derive the constraints~\eqref{eq:rho-output} for each $j = 1,\dots,m_{y_k}$.  The constraint~\eqref{eq:rho-equilibrium} is the condition for $\bar{x}_i$ and $\bar{u}_i \in \U$ to be an equilibrium state and input pair.  
\end{proof}

If the system~\eqref{eq:system} has a unique equilibrium state $\bar{x}_i$ and input $\bar{u}_i$ for the output sample $\bar{y}_i \in \Y$ then problem~\eqref{eq:rho-max} has a closed-form solution as shown in Corollary \ref{cor:rho-lp} below.  

\begin{corollary}
\label{cor:rho-lp}
Let $\bar{x}_i$ and $\bar{u}_i$ be the unique equilibrium state and input~\eqref{eq:equilibrium} respectively for the output $\bar{y}_i \in \Y_k \subseteq \Y$.  Then the optimal solution to problem~\eqref{eq:rho-max} is given by
\begin{align}
\label{eq:rho-closed-form}
\rho_i^\star = \min\left\{ 
\frac{ K_{u}^j - H_{u}^j  \bar{u}_i}{\big\| H_{u}^j F_i P^{-1/2} \big\|}, 
\frac{ K_{y_k}^j - H_{y_k}^j \bar{y}_i }{\big\|  H_{y_k}^j C P^{-1/2} \big\|}
\right\}
\end{align}
where $(H_{u}^j,K_{u}^j)$ for $j = 1,\dots,m_u$ are the half-spaces of the input set~\eqref{eq:input-set} and $(H_{y_k}^j,K_{y_k}^j)$ are the half-spaces of the $k$-th output set~\eqref{eq:output-subsets}.
\end{corollary}

\begin{proof}
Follows directly from~\eqref{eq:rho-input} and~\eqref{eq:rho-output}.  
\end{proof}

Since the Lyapunov matrix $P$ is shared by all the local controllers, the half-space parameters $(H_{u}^j,K_{u}^j)$ and $(H_{y_k}^j,K_{y_k}^j)$ can be normalized offline such that $\| H_{u}^j F_i P^{-1/2} \| = 1$ for all $j = 1,\dots,m_u$ and $\|  H_{y_k}^j C P^{-1/2} \| = 1$ for all $j = 1,\dots,m_{y_k}$.  Thus evaluating~\eqref{eq:rho-closed-form} has computational complexity $O(n_u m_u + n_y m_{y_k})$ where $n_u$ and $n_y$ are the number of system inputs and outputs respectively, and $m_u$ and $m_{y_k}$ are the respective number of constraints that define the input $\U$ and output $\Y_k$ sets.  Evaluating~\eqref{eq:rho-closed-form} has the same computational complexity as testing the set memberships $\bar{u}_i \in \U$ and $\bar{y}_i \in \Y_k$. In fact the computations used to test set membership can be reused to evaluate~\eqref{eq:rho-closed-form}.  Thus the PI sets~\eqref{eq:pi-sets} for the local controllers~\eqref{eq:local-control} and hence the controller graph $\G$ can be constructed efficiently in real-time.

\subsection{Design of Controllers by Semi-Definite Programming}
\label{sec:control-design-sdp}

In this section we present a method for obtaining the local controllers~\eqref{eq:local-control} by solving a semi-definite program.  We assume that the system~\eqref{eq:system} has a unique equilibrium state $\bar{x}_i$ and input $\bar{u}_i$ for each sample output $\bar{y}_i \in \Y$.  

The feedback gain $F_i$ and Lyapunov matrix $P_i$ for the $i$-th controller are obtained by solving the problem
\begin{subequations}
\label{eq:sdp}
\begin{align}
\maximize ~&~ \log \det P_i^{-1}  \label{eq:sdp-cost} 	\\
\subjto 	
~&~ \begin{smatrix} 
P_i^{-1} 					&  \cdot	 \\  
A P_i^{-1} + B F_i P_i^{-1} 	& P_i^{-1} 
\end{smatrix} \succ 0 \label{eq:sdp-invariance} \\
~&~ \begin{smatrix}  P_i^{-1} & \cdot \\ H_{\U}^j F_i P_i^{-1} &  \big(K_{u}^j - H_{u}^j \bar{u}_i\big)^2  \end{smatrix} \succeq 0 \label{eq:sdp-input} \\
~&~ \begin{smatrix} P_i^{-1}  & \cdot \\  H_{y_k}^j C P_i^{-1} & \big(K_{y_k}^j - H_{y_k}^j \bar{y}_i \big)^2 \end{smatrix} \succeq 0. \label{eq:sdp-output} 
\end{align}
\end{subequations}
The constraints of problem~\eqref{eq:sdp} are linear and the cost function~\eqref{eq:sdp-cost} is concave in the decision variables $X = P^{-1}$ and $Y = F P^{-1}$.  Problem~\eqref{eq:sdp} is therefore a semi-definite program, which can be efficiently solved using standard software packages \cite{SeDuMi,SDPT3}.  

Proposition \ref{prop:sdp} shows that problem~\eqref{eq:sdp} finds a controller with the largest constraint satisfying PI set.

\begin{proposition}
\label{prop:sdp}
Problem~\eqref{eq:sdp} finds the local controller~\eqref{eq:local-control} with the largest ellipsoidal PI set~\eqref{eq:pi-sets} that satisfies the input~\eqref{eq:input-set} and output~\eqref{eq:output-set} constraints.  
\end{proposition}

\begin{proof}
First we show that the local controller~\eqref{eq:local-control} where $F_i$ is the solution to~\eqref{eq:sdp} stabilizes the equilibrium $\bar{x}_i$.  The dynamics of the deviation $\tilde{x}_i(t) = x(t) - \bar{x}_i$ of the state $x(t)$ from the equilibrium $\bar{x}_i$ are given by
\begin{align*}
\tilde{x}_i(t+1) 
&= A x(t) + B F_i (x(t) - \bar{x}_i) + B \bar{u}_i - \bar{x}_i \\
&= (A+BF_i) \tilde{x}_i(t)
\end{align*}
where $B \bar{u}_i - \bar{x}_i = A \bar{x}_i$ since $\bar{u}_i$ is the equilibrium input corresponding to the equilibrium state $\bar{x}_i$.  Taking the Schur complement of the constraint~\eqref{eq:sdp-invariance} implies that the matrix $A+BF_i$ satisfies 
\begin{align*}
(A+BF_i)^T P_i (A+BF_i) - P_i \prec 0.
\end{align*}
Thus $A+BF_i$ is Schur and the unit level-set $\OCal_i$ in~\eqref{eq:pi-sets} is positive invariant.  

Next we show that the PI set $\OCal_i$ satisfies output and input constraints.  For a positive definite Lyapunov matrix $P_i \succ 0$, the constraint~\eqref{eq:sdp-output} is equivalent to the matrix inequality 
\begin{align*}
\begin{smatrix} P_i & C^T H_{y_k}^j \\ H_{y_k}^j C & \big( K_{y_k}^j - H_{y_k}^j \bar{y}_i \big)^2 \end{smatrix} \succeq 0.
\end{align*}
Taking the Schur complement produces 
\begin{align*}
P_i \succeq \frac{1}{ \big( K_{y_k}^j - H_{y_k}^j \bar{y}_i \big)^2} C^T H_{y_k}^j (H_{y_k}^j)^T C
\end{align*}
where $K_{y_k}^j - H_{y_k}^j \bar{y}_i \geq 0$ since $\bar{y}_i \in \Y_k$.  Hence
\begin{align*}
H_{y_k}^j(C\tilde{x}_i+\bar{y}_i) = H_{y_k}^jC x \leq K_{y_k}^j
\end{align*} 
whenever $\tilde{x}^T P_i \tilde{x} = (x - \bar{x}_i)^T P_i (x - \bar{x}_i) \leq 1$ where $\bar{y}_i = C \bar{x}_i$ and $\tilde{x}_i = x - \bar{x}_i$.  Therefore $C \OCal_i \subseteq \Y_k$ since $Cx \in \Y_k$ whenever $x \in \OCal_i$.  Likewise we can show $F_i(\OCal_i - \bar{x}_i) + \bar{u}_i \subseteq \U$. 

Finally we note that the volume of the PI set $\OCal_i$ is proportional to $\det P_i^{-1}$.  Thus we conclude that problem \eqref{eq:sdp} finds the local controller \eqref{eq:local-control} with the largest ellipsoidal PI set \eqref{eq:pi-sets} that satisfies the input \eqref{eq:input-set} and output \eqref{eq:output-subsets} constraints.  
\end{proof}

Designing both the local controllers~\eqref{eq:local-control} and the PI sets~\eqref{eq:pi-sets} for each output $y_i \in \Y$ using the semi-definite program~\eqref{eq:sdp} can be time consuming.  The computational burden can be eased by fixing the feedback gain matrix $F_i$ of the local controller~\eqref{eq:local-control} and only designing the PI sets.  In this case, the simplified SDP is given by 
\begin{subequations}
\label{eq:sdpr}
\begin{align}
\max 	~&~ \log \det P_i^{-1} \\
\mathrm{s.t.} 		~&~ P_i^{-1} - (A+BF_i) P_i^{-1} (A+BF_i)^T  \succeq 0 \label{eq:sdpr-invariance} \\
				~&~ H_u^J F_i P_i^{-1} F_i^T (H_u^j)^T \leq (K_u^j - H_u^j \bar{u}_i )^2 \label{eq:sdpr-input} \\
				~&~ H_{y_k}^J C P_i^{-1} C^T (H_{y_k}^j)^T \leq (K_{y_k}^j - H_{y_k}^j \bar{y}_i )^2 \label{eq:sdpr-output}.
\end{align}
\end{subequations}
This problem is a semi-definite program in the decision variable $P_i^{-1}$.  The following corollary shows that~\eqref{eq:sdpr} is equivalent to~\eqref{eq:sdp} with a fixed feedback gain $F_i$.  

\begin{corollary}
Let $F_i$ be fixed.  Then the semi-definite program~\eqref{eq:sdp} is equivalent to~\eqref{eq:sdpr}.  
\end{corollary}
 
\begin{proof}
Taking the Schur complement of the constraint~\eqref{eq:sdpr-invariance} and multiplying by $\mathrm{diag}(I,P_i^{-1})$ produces~\eqref{eq:sdp-invariance}.  The constraints~\eqref{eq:sdpr-input} and~\eqref{eq:sdpr-output} are the Schur complement of~\eqref{eq:sdp-input} and~\eqref{eq:sdp-output} respectively.  
\end{proof}

\section{Example: Spacecraft Maneuver Planning}
\label{sec:example}

In this section we apply our path-planning algorithm to the problem of planning spacecraft docking maneuvers.  The relative dynamics of a pair of spacecraft in the orbital-plane are modeled by the Hill-Clohessy-Wiltshire equations \cite{Wie2010} 
\begin{subequations}
\label{eq:spacecraft-dynamics}
\begin{align}
\ddot{y}_1 &= 2n\dot{y}_2 + 3 n^2 y_1 + u_1 \\
\ddot{y}_2 &= -2n\dot{y}_1 + u_2 
\end{align}
\end{subequations}
where $y_1$ is the difference in radial position of the spacecraft and $y_2$ is the difference in position along the orbital velocity direction.  The state vector $x = [y_1,y_2,\dot{y_1},\dot{y_2}]^T$ contains the relative radial and orbital positions and velocities.  The inputs $u_1$ and $u_2$ are the thrusts normalized by the spacecraft mass along the radial and orbital velocity directions.  The dynamics are linearized about a circular orbit of $415$ kilometers which gives $n = 1.1\times 10^{-3}$ inverse-seconds in~\eqref{eq:spacecraft-dynamics}.  The dynamics~\eqref{eq:spacecraft-dynamics} are discretized with a sample period of $30$ seconds.  The normalized thrusts must satisfy the input constraints 
\begin{align}
\label{eq:spacecraft-input}
-10^{-2} \leq u_1,u_2 \leq 10^{-2}
\end{align}
newtons per kilogram.

We consider the scenario of planning a maneuver around a piece of debris shown in Fig.~\ref{fig:spacecraft}.  The debris is a square with a $100$ meter side length located at $[300,400]^T$ meters.  The output set $\Y$ is the set difference of the bounding box $[-400,1000] \times [-400,1100]$ meters and the debris set.  The component sets $\Y_1,\dots,\Y_4$ that comprise the output set $\Y = \Y_1\cup \cdots \cup \Y_4$ were obtained by flipping each of the $4$ constraints that define the debris set and intersecting with the bounding box.   The component sets $\Y_1,\dots,\Y_4$ are shown in Fig.~\ref{fig:Y}.

\begin{figure}[h]
\centering
\includegraphics[width=0.5\columnwidth]{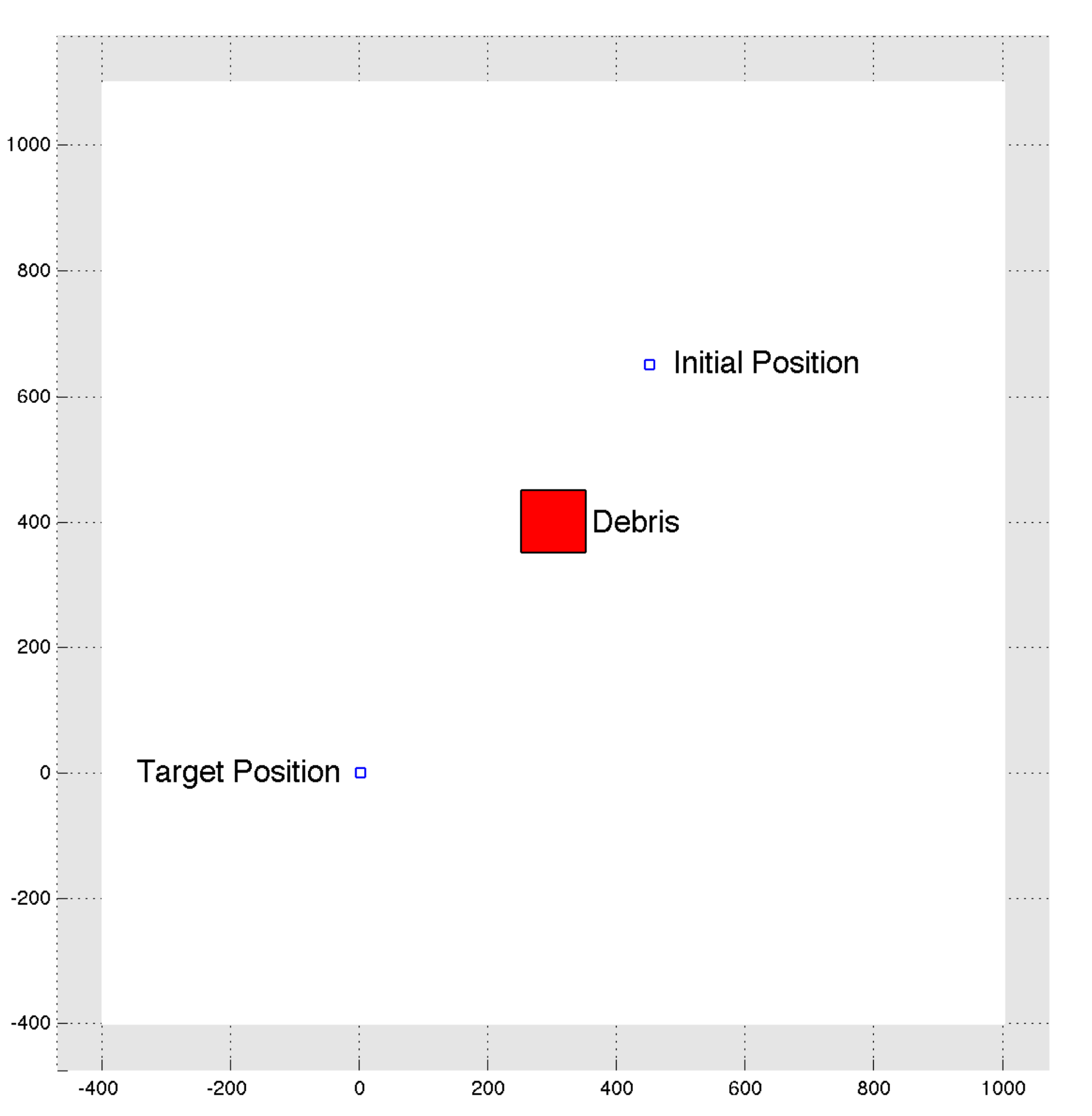}
\caption{Initial position and target position of the spacecraft with debris between.  The white area is the output set $\Y \subset \Re^2$ and the red square is the debris.}
\label{fig:spacecraft}
\end{figure}

\begin{figure}[h]
\centering
\subfigure[$\Y_1$]{
\includegraphics[width=0.2\columnwidth]{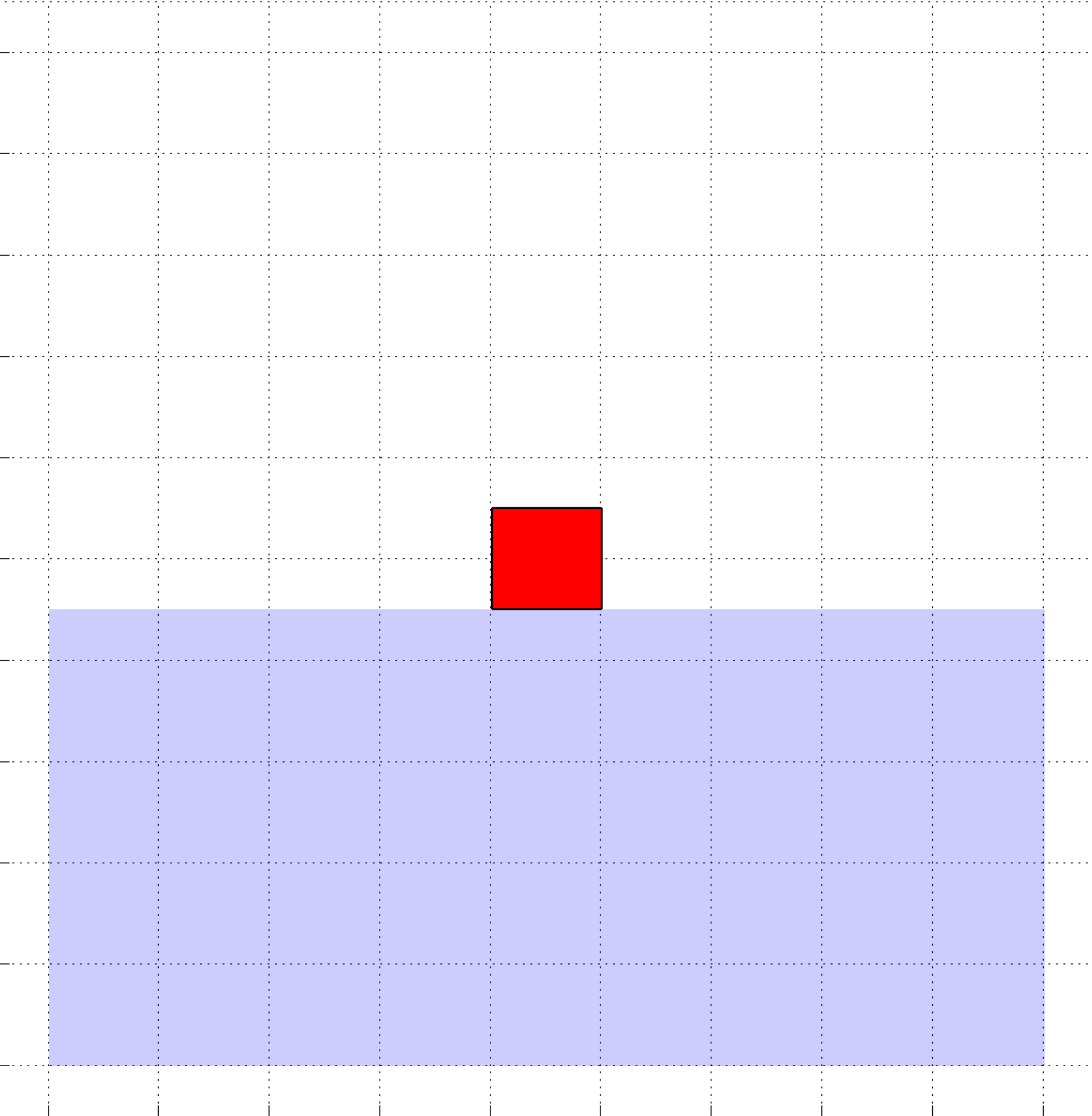}
}
\subfigure[$\Y_2$]{
\includegraphics[width=0.2\columnwidth]{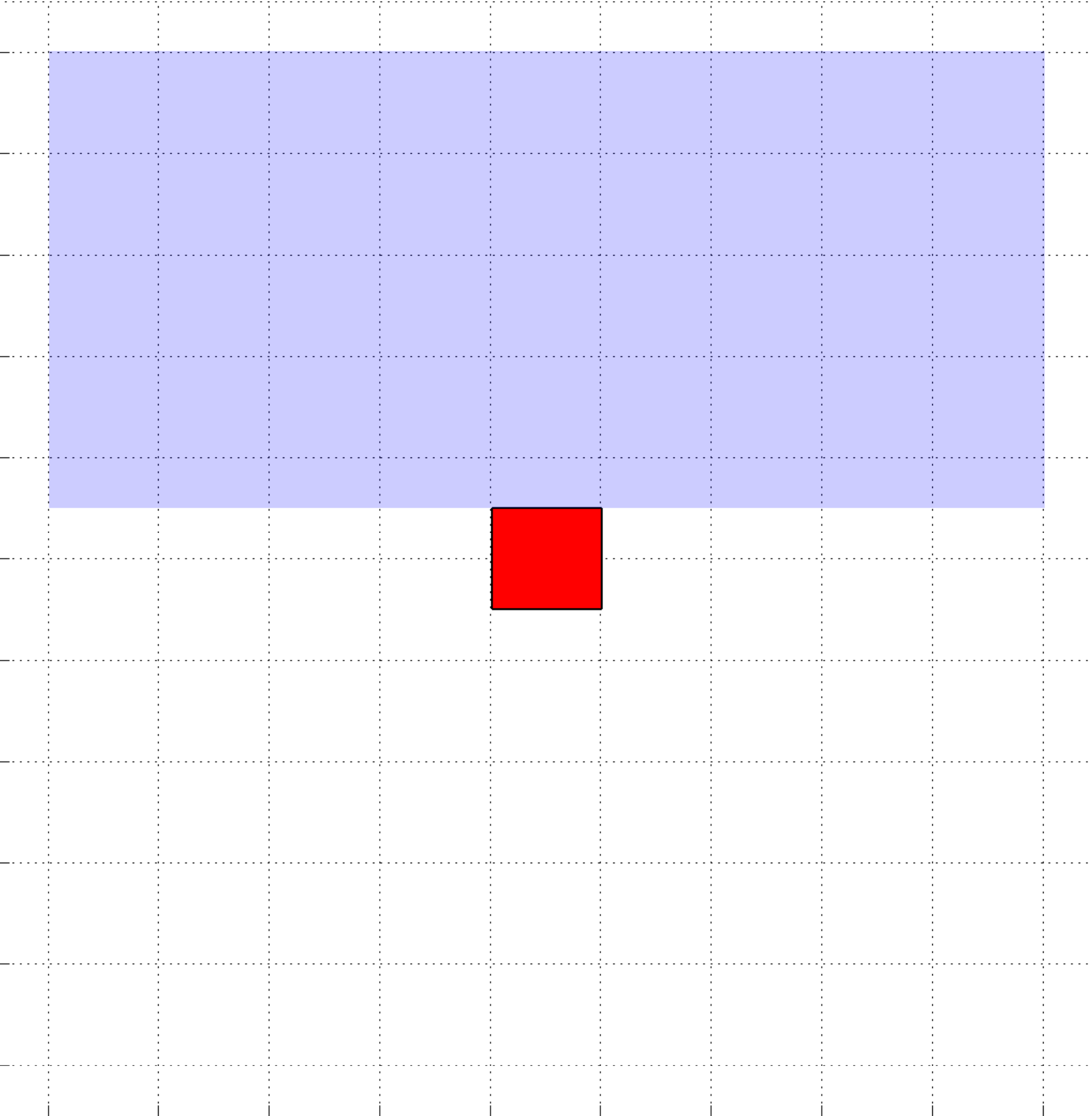}
}
\subfigure[$\Y_3$]{
\includegraphics[width=0.2\columnwidth]{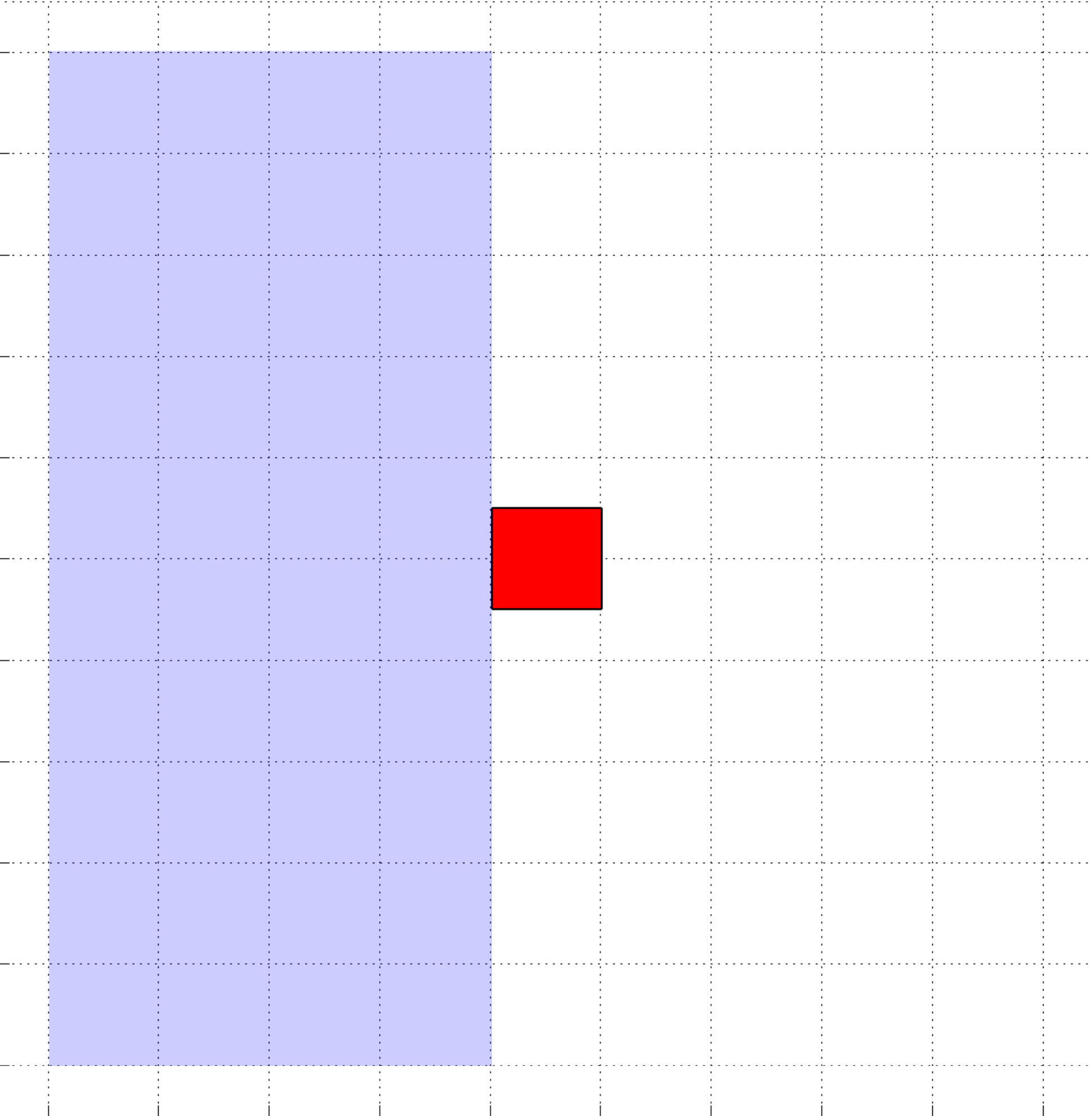}
}
\subfigure[$\Y_4$]{
\includegraphics[width=0.2\columnwidth]{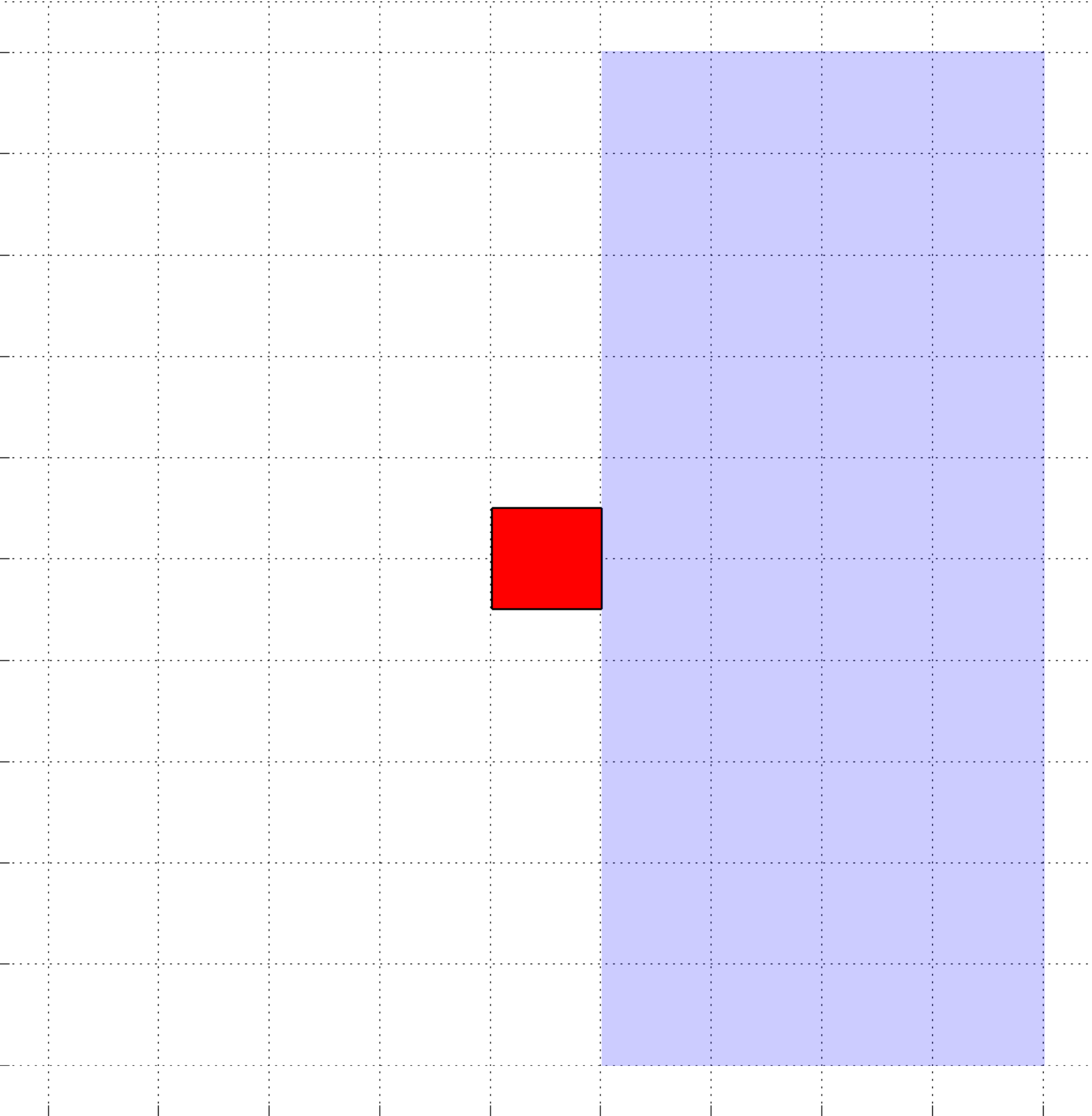}
}
\caption{The four component sets $\Y_1,\Y_2,\Y_3,\Y_4$ that comprise the non-convex output set $\Y = \Y_1 \cup\Y_2 \cup\Y_3 \cup \Y_4$. }
\label{fig:Y}
\end{figure}

The spacecraft is initially positioned at $y = [450,650]^T$ meters and the target position is the origin $y_f = [0,0]^T$.  The controller graph $\G = (\I,\E)$ was generated by gridding the output set $\mathrm{conv}(\Y)$ and computing a local controller~\eqref{eq:local-control} at each grid point $\bar{y}_i \in \Y$ using the methods presented in Section~\ref{sec:control-design-fixed} and Section~\ref{sec:control-design-sdp}.  For the first controller graph, the feedback gain $F = F_i$ is the linear quadratic regulator (LQR) with penalty matrices $Q = \mathrm{diag}(10^2,10^2,10^7,10^7)$ and  $R = (2\times10^7) I_2$ where $\mathrm{diag}(v)$ is a diagonal matrix with elements $v$ on the diagonal and $I_2 \in \Re^{2 \times 2}$ is the identity matrix.  The Lyapunov matrix $P \in \Re^{n_x \times n_x}$ is the corresponding solution to the discrete-time algebraic Riccati equation.  For the second controller graph, the feedback gains $F_i$ and Lyapunov matrices $P_i$ of the local controllers~\eqref{eq:local-control} were obtained by solving the semi-definite program~\eqref{eq:sdp} at each grid point.  

The projection $C\OCal_i$ of the PI sets $\OCal_i$ and the controller graph $\G$ for the fixed-gain controller are shown in Fig.~\ref{fig:pi-fixed} and Fig.~\ref{fig:graph-fixed} respectively.  The projected PI sets $C\OCal_i$ for the local controllers designed by the semi-definite program~\eqref{eq:sdp} are shown in Fig.~\ref{fig:pi-sdp}.  The controller graph is shown in Fig.~\ref{fig:graph-sdp}.  

From Fig.~\ref{fig:pi-fixed} and Fig.~\ref{fig:pi-sdp} it is evident that the PI sets for the controllers designed using the semi-definite program~\eqref{eq:sdp} are larger than the scaled PI sets for the fixed-gain controller.  As a result the controller graph shown in Fig.~\ref{fig:graph-sdp} has more edges than the controller graph for the fixed-gain controller shown in Fig.~\ref{fig:graph-fixed}.  Thus the path-planning algorithm has a larger number of potential sequences of local controllers that can maneuver the spacecraft to the origin.  Hence we expect the path planner using the SDP controllers to perform better than the fixed-gain controller.  

\begin{figure}[h]
\centering
\subfigure[]{
\includegraphics[width=0.35\columnwidth]{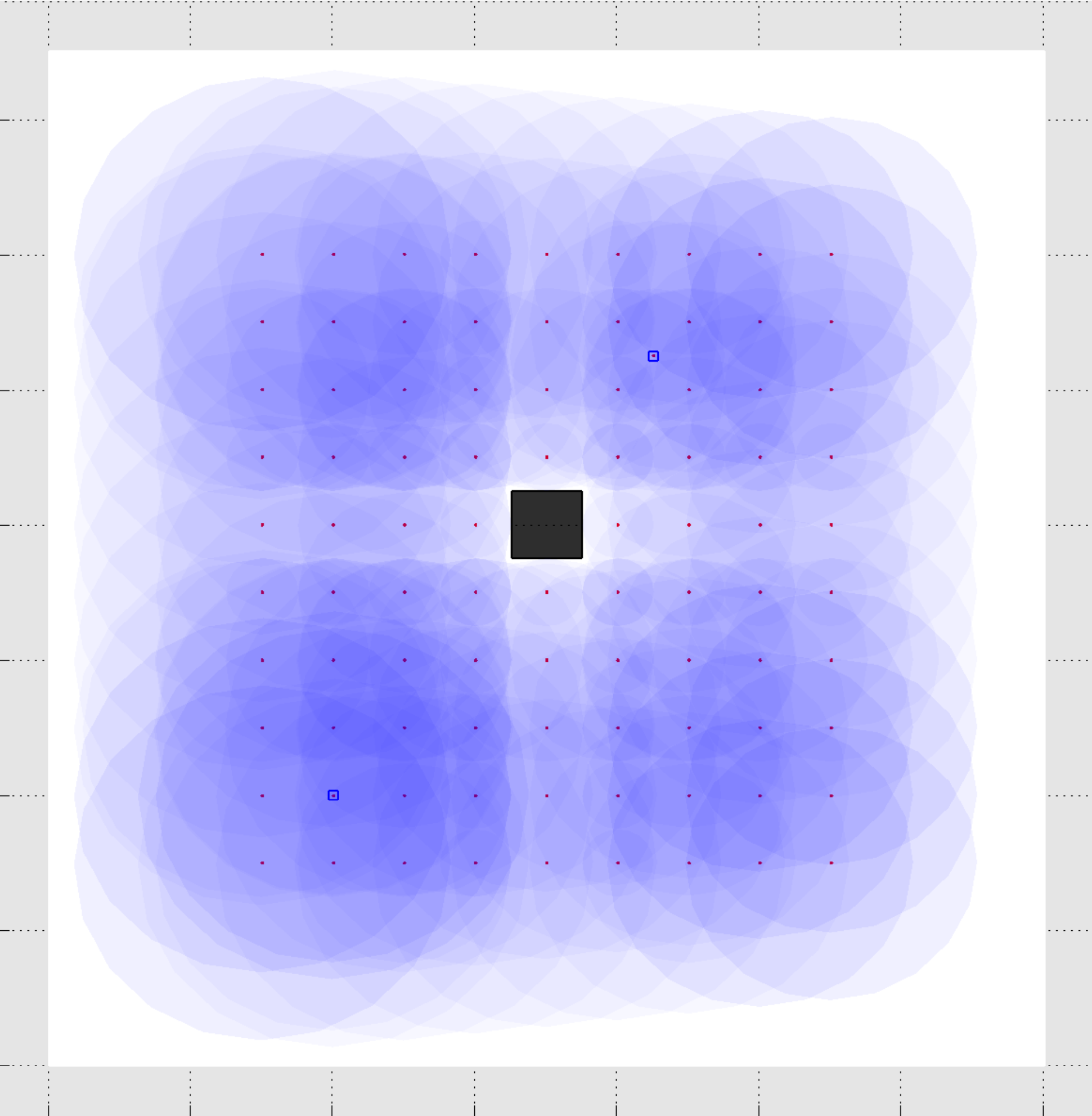}
\label{fig:pi-fixed}
} 
\subfigure[]{
\includegraphics[width=0.35\columnwidth]{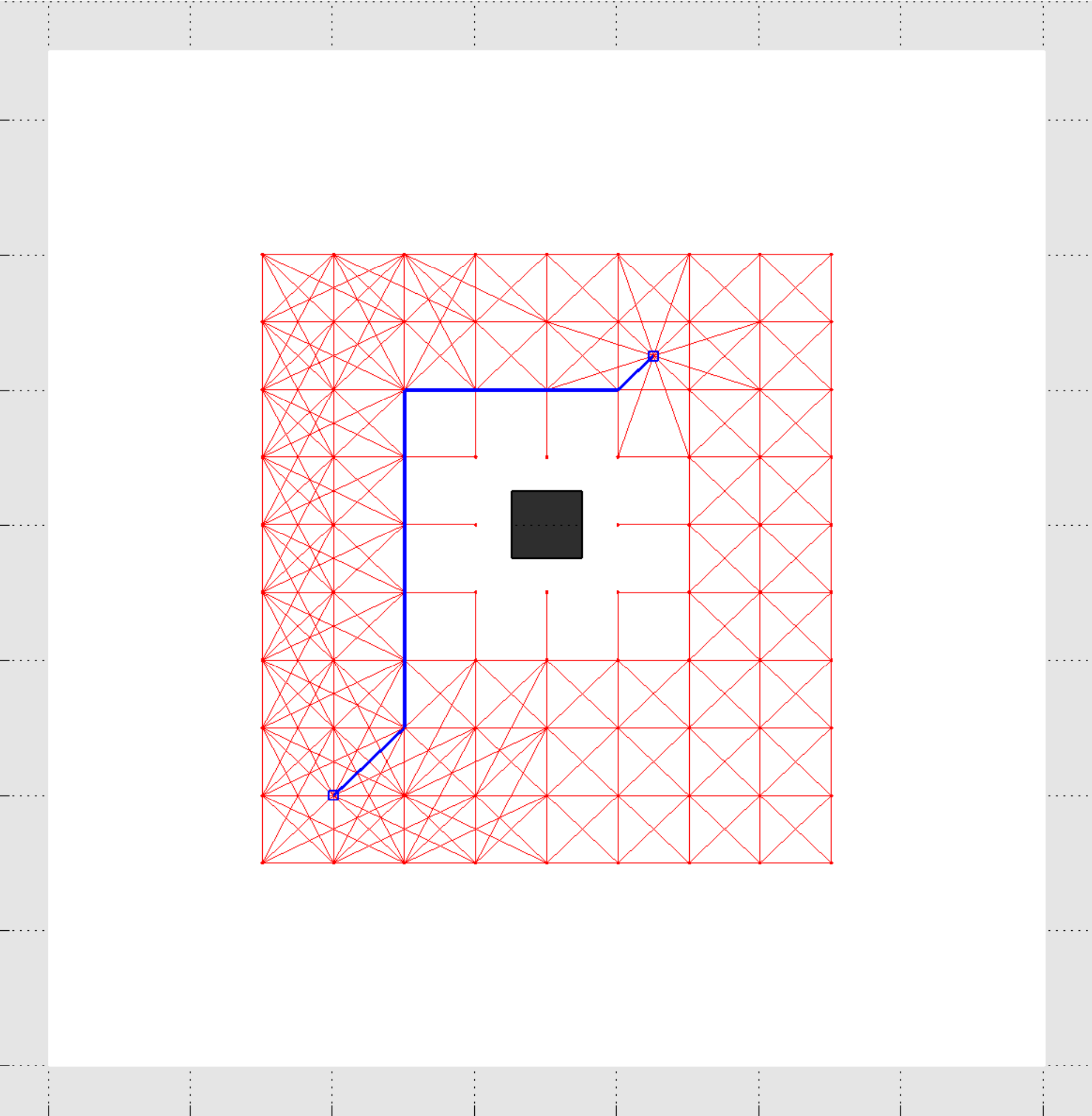}
\label{fig:graph-fixed}
} 
\subfigure[]{
\includegraphics[width=0.35\columnwidth]{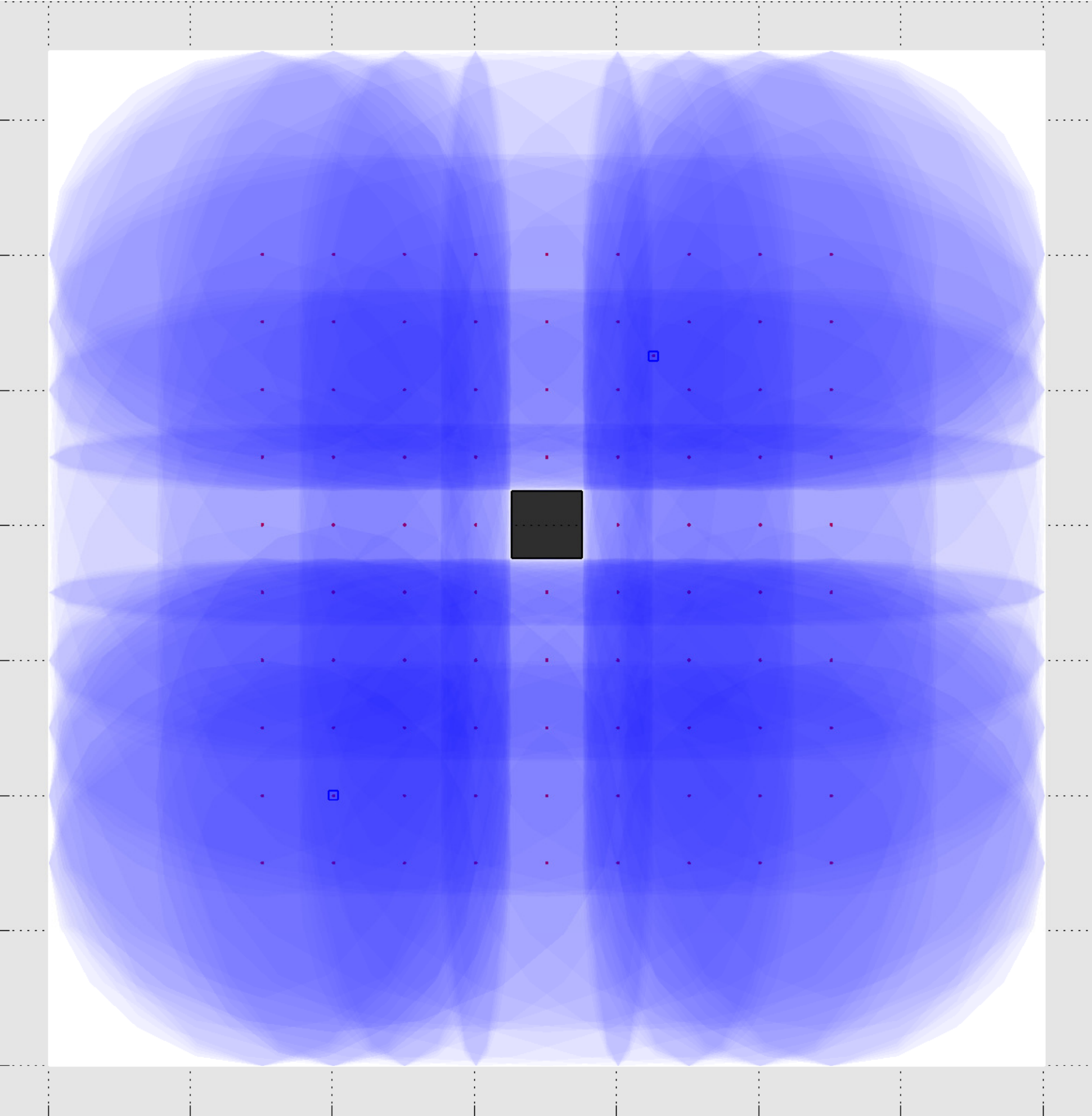}
\label{fig:pi-sdp}
} 
\subfigure[]{
\includegraphics[width=0.35\columnwidth]{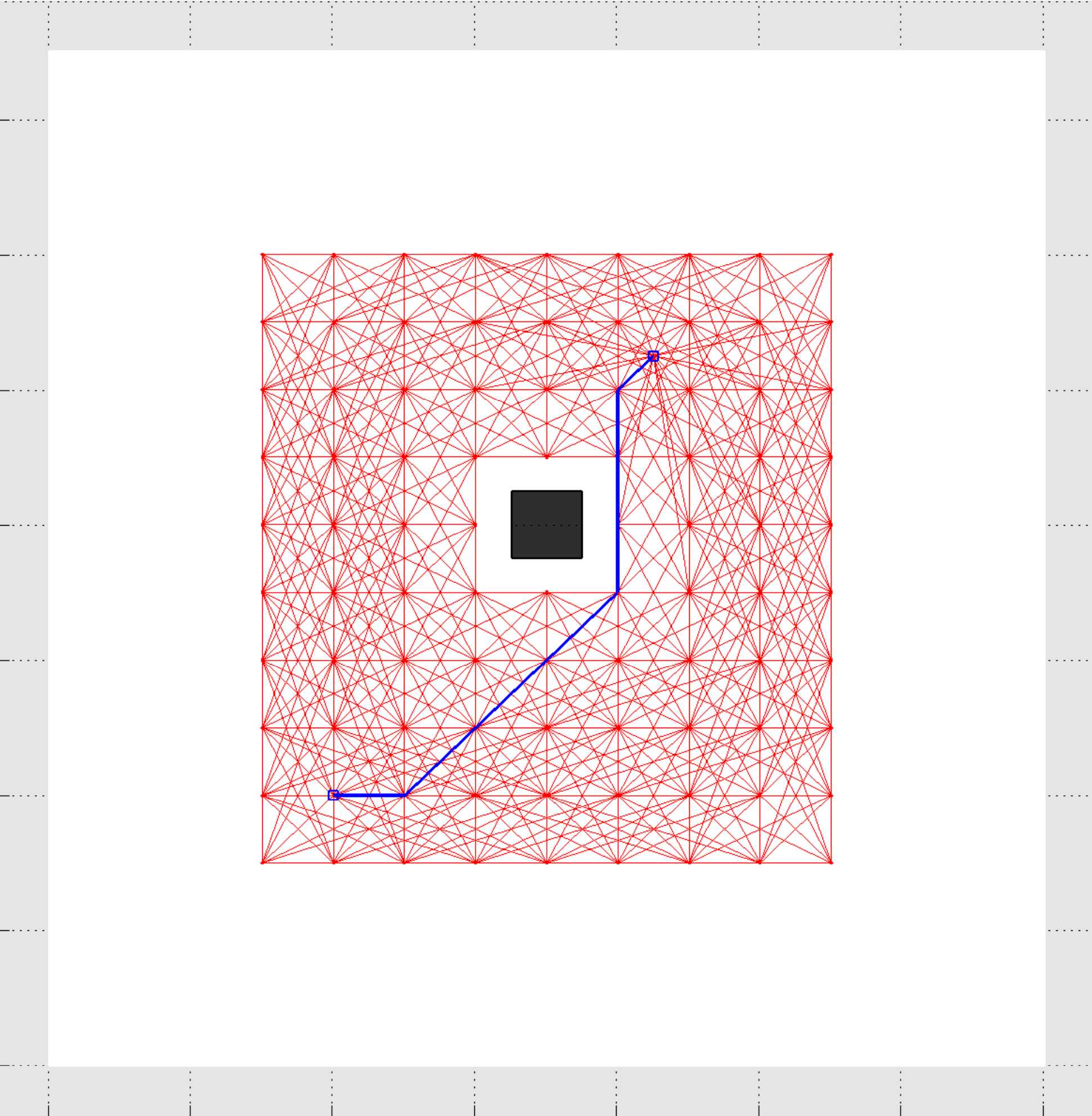}
\label{fig:graph-sdp}
}
\caption{(a) and (c) Projected PI sets $C \OCal_i$ for the local controllers designed using a fixed-gain controller and SDP controller respectively. (b) and (d) Controller graphs $\G = (\I,\E)$ for the fixed-gain and SDP controllers.}
\end{figure}

The edges $(i,j) \in \E$ of the controller graphs $\G = (\I,\E)$ are weighted $W_{ij}$ using the infinite-horizon cost-to-go from initial state $\bar{x}_i$ to the equilibrium state $\bar{x}_j$ under the local controller~\eqref{eq:local-control} given by
\begin{align*}
W_{ij} = (\bar{x}_i - \bar{x}_j)^T S_j (\bar{x}_i - \bar{x}_j)
\end{align*}
where $S_j = S$ is the solution to the discrete-time algebraic Riccati equation for the fixed-gain controller and $S_j$ is the unique positive definite solution to the Lyapunov equation
\begin{align*}
(A+BF_j)^T S_j (A+BF_j) - S_j = -Q - F_j^T R F_j.
\end{align*}
for the controller $F_j$ designed using the semi-definite program~\eqref{eq:sdp}.  

The optimal sequence of local controller nodes are shown in Fig.~\ref{fig:graph-fixed} and Fig.~\ref{fig:graph-sdp} respectively.   These controller sequences are used by Algorithm \ref{alg:path-planning} to maneuver the spacecraft to the origin.  The resulting output trajectories are shown in Fig.~\ref{fig:trajectories}.  The output trajectories are compared with the output trajectory produced by using a single LQR.  Under the single LQR, the spacecraft passed through the debris field before converging to the target position.   On the other hand, the output trajectories produced by Algorithm \ref{alg:path-planning} using the fixed-gain and SDP local controllers avoided the debris set while converging to the target position.  However the output trajectories took very different paths to reach the origin.  We evaluated these trajectories using the cost function cost-to-go to a neighborhood of the origin
\begin{align*}
J = \displaystyle{\sum}_{t=0}^{N} x(t)^T Q x(t) + u(t)^T R u(t)
\end{align*}
where $x(t)$ and $u(t)$ were the state and input trajectories, respectively, produced by Algorithm \ref{alg:path-planning}.  The fixed-gain local controllers had a cost $J_{base} = 1.14 \times 10^{10}$ and the SDP designed local controllers had a cost $J_{sdp} = 2.15 \times 10^9$.  Thus for this example problem the SDP designed local controller produced a more efficient path to the origin than the fixed-gain local controllers.  The disadvantage of the SDP-based control design is computation time.  The PI sets and controller graph for the fixed-gain controller required $0.36$ seconds to compute.  While the PI sets and controller graph for the SDP based controllers required $48$ seconds to compute.  Thus the SDP based control design required more than $2$ orders of magnitude more time to compute while providing approximately one order of magnitude improvement in performance.  

\begin{figure}[h]
\centering
\includegraphics[width=0.6\columnwidth]{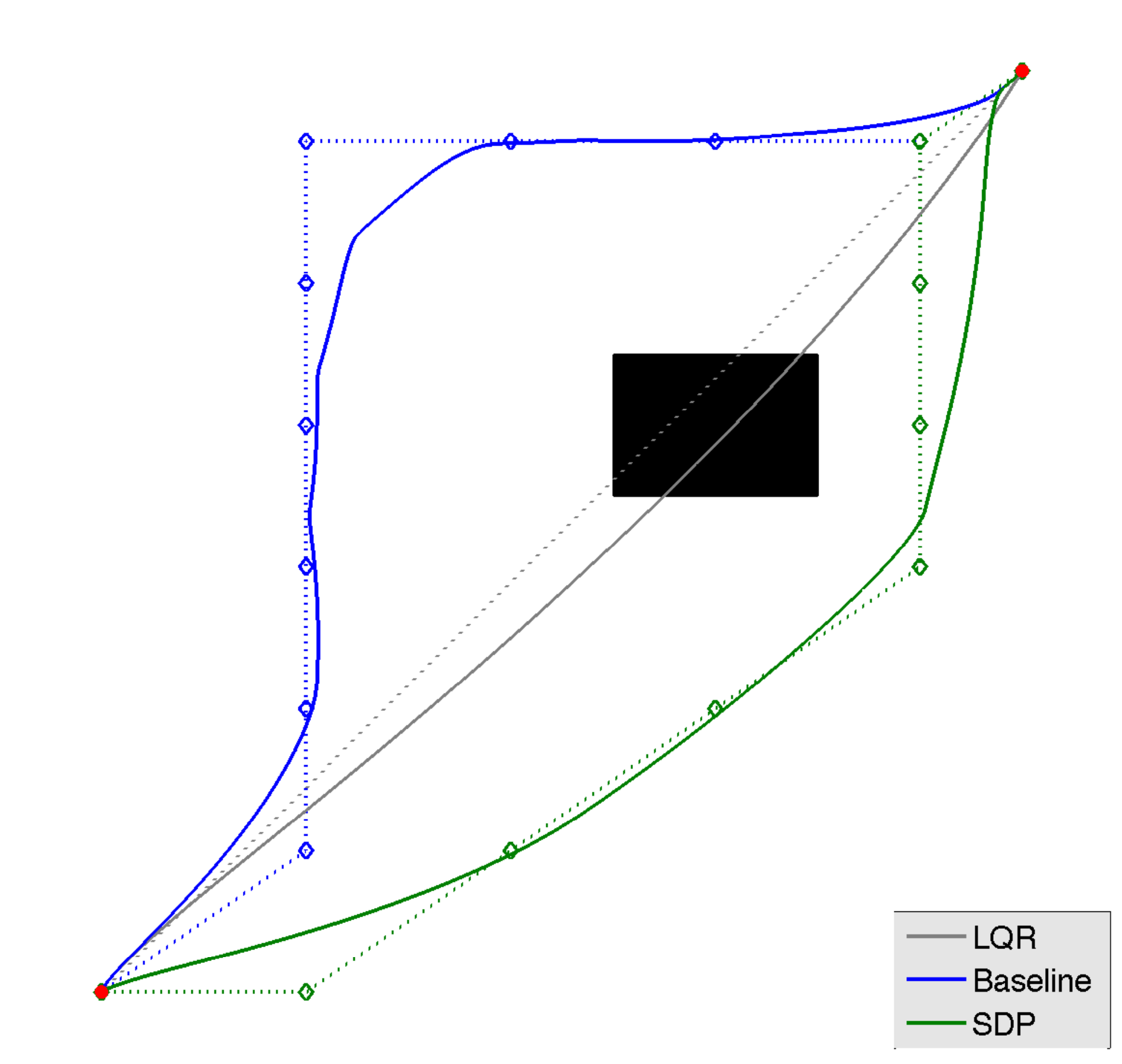}
\caption{Trajectories produced by Algorithm \ref{alg:path-planning} by searching the controller graph $\G = (\I,\E)$ and using a sequence of local controllers.  Shown are the sequences of equilibrium outputs $\bar{y}_{\sigma_1},\dots,\bar{y}_{\sigma_f}$ from the controller graph $\G$ and the resulting output trajectory $y(t)$ produced by switching between these equilibrium outputs.}
\label{fig:trajectories}
\end{figure}

The normalized thrusts for both controller graphs and the LQR controller are shown in Fig.~\ref{fig:input}.  The LQR violated the thrust constraints while the inputs $u(t)$ produced by Algorithm \ref{alg:path-planning} did not violate constraints.  Overall the normalized thrusts produced by the path-planning algorithm are small, hence requiring little propellant.  

\begin{figure}[h]
\centering
\includegraphics[width=0.75\columnwidth]{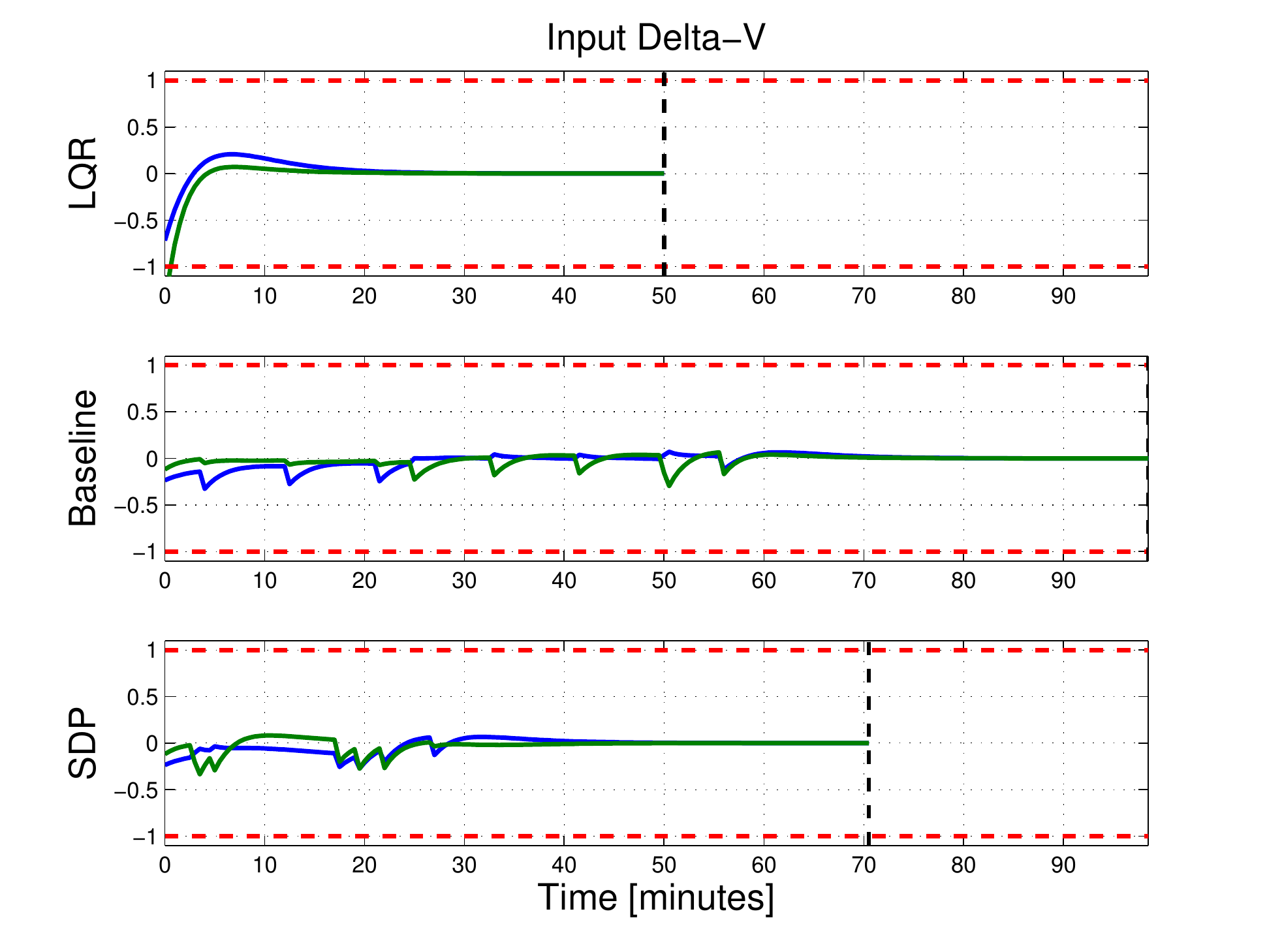}
\caption{Input trajectories $u(t)$ for $t \in \N$ produced by the linear quadratic regulator, path planner with the fixed-gain local controller, and path planner with the SDP local controllers.  The LQR controller violates input constraints.  The fixed-gain and SDP controllers satisfy input constraints.}
\label{fig:input}
\end{figure}

\section{Conclusions and Future Work}

This paper presented a path-planning algorithm for linear systems whose output is restricted to a union of polytopic sets.  The path-planning algorithm uses a graph to switch between a collection of local linear feedback controllers.  We presented two methods for computing the local controllers and their positive invariant sets.  The first used a fixed feedback gain and scaled the positive invariant set to satisfy state and input constraints.  The second used a semi-definite program to design both the controller gain and positive invariant set.  The path-planning algorithm was applied to the spacecraft docking problem.

Future work will address path-planning for systems with disturbances and modeling errors.  In addition we will study the problem of sampling the output space $\Y$ in a manner that guarantees that the controller graph $\G$ contains a path from the current output to the target output when one exist.  Finally we will study different weighting heuristics for the controller graph edges and how they effect closed-loop performance.  

\bibliographystyle{IEEEtran}
\bibliography{path-planning.bib}

\begin{thebibliography}{10}
\providecommand{\url}[1]{#1}
\csname url@samestyle\endcsname
\providecommand{\newblock}{\relax}
\providecommand{\bibinfo}[2]{#2}
\providecommand{\BIBentrySTDinterwordspacing}{\spaceskip=0pt\relax}
\providecommand{\BIBentryALTinterwordstretchfactor}{4}
\providecommand{\BIBentryALTinterwordspacing}{\spaceskip=\fontdimen2\font plus
\BIBentryALTinterwordstretchfactor\fontdimen3\font minus
  \fontdimen4\font\relax}
\providecommand{\BIBforeignlanguage}[2]{{%
\expandafter\ifx\csname l@#1\endcsname\relax
\typeout{** WARNING: IEEEtran.bst: No hyphenation pattern has been}%
\typeout{** loaded for the language `#1'. Using the pattern for}%
\typeout{** the default language instead.}%
\else
\language=\csname l@#1\endcsname
\fi
#2}}
\providecommand{\BIBdecl}{\relax}
\BIBdecl

\bibitem{reif1979complexity}
J.~H. Reif, ``Complexity of the mover's problem and generalizations,'' in
  \emph{Conference on Foundations of Computer Science}, 1979.

\bibitem{lavalle2006planning}
S.~LaValle, \emph{Planning Algorithms}.\hskip 1em plus 0.5em minus 0.4em\relax
  Cambridge University Press, 2006.

\bibitem{lavalle2001randomized}
S.~LaValle and J.~Kuffner, ``Randomized kinodynamic planning,'' \emph{The
  International Journal of Robotics Research}, 2001.

\bibitem{karaman2011sampling}
S.~Karaman and E.~Frazzoli, ``Sampling-based algorithms for optimal motion
  planning,'' \emph{The International Journal of Robotics Research}, 2011.

\bibitem{kuffner2002dynamically}
J.~Kuffner, S.~Kagami, K.~Nishiwaki, M.~Inaba, and H.~Inoue,
  ``Dynamically-stable motion planning for humanoid robots,'' \emph{Autonomous
  Robots}, 2002.

\bibitem{leonard2008perception}
J.~Leonard, J.~How, S.~Teller, M.~Berger, S.~Campbell, G.~Fiore, L.~Fletcher,
  E.~Frazzoli, A.~Huang, S.~Karaman \emph{et~al.}, ``A perception-driven
  autonomous urban vehicle,'' \emph{Journal of Field Robotics}, 2008.

\bibitem{perez2011asymptotically}
A.~Perez, S.~Karaman, A.~Shkolnik, E.~Frazzoli, S.~Teller, and M.~Walter,
  ``Asymptotically-optimal path planning for manipulation using incremental
  sampling-based algorithms,'' in \emph{Intelligent Robots and Systems (IROS)},
  2011.

\bibitem{frazzoli2003quasi}
E.~Frazzoli, ``Quasi-random algorithms for real-time spacecraft motion planning
  and coordination,'' \emph{Acta Astronautica}, 2003.

\bibitem{starek2016real}
J.~Starek, G.~Maher, B.~Barbee, and M.~Pavone, ``Real-time, fuel-optimal
  spacecraft motion planning under {Clohessy-Wiltshire-Hill} dynamics,'' in
  \emph{IEEE Aerospace Conference}, 2016.

\bibitem{vinter2010optimal}
R.~Vinter, \emph{Optimal control}.\hskip 1em plus 0.5em minus 0.4em\relax
  Springer Science, 2010.

\bibitem{Weiss2014}
A.~Weiss, C.~Petersen, M.~Baldwin, R.~Erwin, and I.~Kolmanovsky, ``Safe
  positively invariant sets for spacecraft obstacle avoidance,'' \emph{Journal
  of Guidance, Control, and Dynamics}, 2015.

\bibitem{Kolmanovsky1998}
I.~Kolmanovsky and E.~G. Gilbert, ``Theory and computation of disturbance
  invariant sets for discrete-time linear systems,'' \emph{Mathematical
  problems in engineering}, 1998.

\bibitem{Kerrigan2000}
E.~Kerrigan, ``Robust constraints satisfaction: Invariant sets and predictive
  control,'' Ph.D. dissertation, Dep. of Engineering, University of Cambridge,
  2000.

\bibitem{arslan2016sampling}
O.~Arslan, K.~Berntorp, and P.~Tsiotras, ``Sampling-based algorithms for
  optimal motion planning using closed-loop prediction,'' \emph{arXiv preprint
  arXiv:1601.06326}, 2016.

\bibitem{McConley2000}
W.~McConley, B.~Appleby, M.~Dahleh, and E.~Feron, ``A computationally efficient
  lyapunov-based scheduling procedure for control of nonlinear systems with
  stability guarantees,'' \emph{Transactions on Automatic Control}, 2000.

\bibitem{Blanchini2004}
F.~Blanchini, F.~Pellegrino, and L.~Visentini, ``Control of manipulators in a
  constrained workspace by means of linked invariant sets,'' \emph{Journal of
  Robust and Nonlinear Control}, 2004.

\bibitem{FullPaper}
C.~Danielson, A.~Weiss, K.~Berntorp, and S.~{Di Cairano}, ``Path planning using
  positive invariant sets,'' \emph{arXiv preprint arXiv:TBD}, 2016.

\bibitem{Borrelli2009}
F.~Borrelli, A.~Bemporad, and M.~Morari, \emph{Predictive Control for Linear
  and Hybrid Systems}, 2012.

\bibitem{SeDuMi}
J.~Sturm, ``Using {SeDuMi} 1.02, a {MATLAB} toolbox for optimization over
  symmetric cones,'' \emph{Optimization Methods and Software}, 1999.

\bibitem{SDPT3}
K.~Toh, M.~Todd, and R.~Tutuncu, ``{SDPT3} --- a {MATLAB} software package for
  semidefinite programming,'' \emph{Optimization Methods and Software}, 1999.

\bibitem{Wie2010}
B.~Wie, \emph{Spacecraft Dynamics and Control}.\hskip 1em plus 0.5em minus
  0.4em\relax AIAA, 2010.

\end{thebibliography}

\end{document}